\documentclass[11pt]{amsart}

\usepackage{amsmath,amssymb,amsthm}
\usepackage{enumitem}
\usepackage{tikz}
\usetikzlibrary{matrix}
\usetikzlibrary{arrows}
\usetikzlibrary{calc}
\usetikzlibrary{backgrounds}
\usetikzlibrary{automata}
\usetikzlibrary{positioning}
\usetikzlibrary{decorations.pathreplacing}

\usepackage{verbatim}

\usepackage{setspace}
\theoremstyle{plain}
\newtheorem{theorem}{Theorem}
\newtheorem{proposition}[theorem]{Proposition}
\newtheorem{lemma}[theorem]{Lemma}
\newtheorem{corollary}[theorem]{Corollary}

\theoremstyle{definition}

\newtheorem{example}[theorem]{Example}
\newtheorem{remark}[theorem]{Remark}

\newtheorem*{convention}{Convention}

\newcommand{\onto}{\twoheadrightarrow}

\newcommand{\into}{\hookrightarrow}

\renewcommand{\theta}{\vartheta}
\renewcommand{\phi}{\varphi}

\newcommand{\emp}{\varepsilon}

\newcommand{\FA}{\widehat{F}_{\mathbf{A}}}
\newcommand{\FAf}{F_{\mathbf{A}}}
\newcommand{\TA}{T_A}
\newcommand{\TAf}{T^\mathrm{fin}_A}
\newcommand{\LT}{\mathrm{LT}}
\newcommand{\AL}{\Lambda(A)}
\newcommand{\BL}{\Lambda(B)}
\newcommand{\Cont}{\mathcal{C}}
\newcommand{\Reg}{\mathrm{Reg}}
\newcommand{\RT}{\mathrm{RT}}
\newcommand{\CT}{\mathrm{CT}}
\newcommand{\J}{\mathcal{J}}
\newcommand{\R}{\mathcal{R}}
\renewcommand{\L}{\mathcal{L}}
\newcommand{\Stab}{\mathrm{Stab}}
\newcommand{\sem}[1]{[\![#1]\!]}
\DeclareMathOperator{\St}{\mathrm{Spec}}
\DeclareMathOperator{\K}{\mathcal{K}}
\newcommand{\be}[1]{({<}#1)} 
\newcommand{\af}[1]{({>}#1)} 
\newcommand{\op}{\mathrm{op}}
\renewcommand{\hat}{\widehat}
\newcommand{\x}{\overline{x}}

\numberwithin{theorem}{section}

\newcommand{\labeldist}{.5}

\newcommand\drawwordn[3]{
\begin{scope}[shift={#1}]
    \draw [thick] (0,0) -- (#2,0);
    \draw [thick] (0,-.25) -- (0,.25);
    \draw [thick] (#2,-.25) -- (#2,.25);
    \node [above=\labeldist] at (.5*#2,0) {#3};
\end{scope}
}

\newcommand\drawwordnr[3]{
\begin{scope}[shift={#1}]
    \draw [thick] (0,0) -- (#2,0);
    \draw [thick] (0,-.25) -- (0,.25);
    \draw [thick] (#2,-.25) -- (#2,.25);
    \node [right=\labeldist] at (#2,0) {#3};
\end{scope}
}

\newcommand\drawwordo[4]{
\begin{scope}[shift={#1}]
    \draw [thick,#4] (0,0) -- (#2,0);
    \draw [thick,#4] (0,-.25) -- (0,.25);
    \draw [thick,#4] (#2,-.25) -- (#2,.25);
    \node [above=\labeldist,#4] at (.5*#2,0) {#3};
\end{scope}
}

\newcommand{\drawdecor}[4]{
  \draw [decorate,line width=1pt,#4] #1 -- node[below=0.1] {#3} #2;
}

\newcommand\drawsubtri[4]{
 \draw[fill=black] #1 circle (2pt);
  \node[above=\labeldist] at #1 {#4};
  \draw [thick,dotted] #1 -- #2;
  \draw [thick,dotted] #1 -- #3;
}

\title{Pro-aperiodic monoids and model theory}
\author{Samuel J. v. Gool and Benjamin Steinberg}

\begin{document}

\begin{abstract}
We apply Stone duality and model theory to study the structure theory of free pro-aperiodic monoids. Stone duality implies that elements of the free pro-aperiodic monoid may be viewed as elementary equivalence classes of pseudofinite words. Model theory provides us with saturated words in each such class, i.e., words in which all possible factorizations are realized. We give several applications of this new approach, including a solution to the word problem for $\omega$-terms that avoids using McCammond's normal forms, as well as new proofs and extensions of other structural results concerning free pro-aperiodic monoids. 
\end{abstract}

\maketitle

\section*{Introduction}
The pseudovariety of aperiodic monoids has long played a fundamental role in finite semigroup theory and automata theory.  The famous Sch\"utz\-en\-ber\-ger theorem~\cite{Sch1965} proved that the aperiodic monoids recognize precisely the star-free languages; this class was later shown by McNaughton and Papert~\cite{McNPap1971} to be the class of first-order definable languages; also see Straubing's book \cite{Str1994}.  From the point of view of semigroup decomposition theory, aperiodic monoids are important because they form one of the basic building blocks, along with the pseudovariety of finite groups, of all finite semigroups according to the Krohn-Rhodes decomposition theorem~\cite{KR65}.  Thus aperiodic monoids play a prime role in two of the oldest open problems in automata theory: the dot-depth problem \cite{Brz1976,PlaZei2014,Pin2016} and the Krohn-Rhodes complexity problem~\cite{KR68}.

The importance of profinite monoids in automata theory and finite semigroup theory was first highlighted, starting in the late eighties, by Almeida; see his influential book~\cite{Almeida:book}, and the more recent monograph~\cite{RS2009}, for more background.  Much of the early work focused on viewing elements of relatively free profinite monoids as limits of finite words, but early on Almeida made the important observation that the Boolean algebra of clopen subsets of the free pro-$\mathbf V$ monoid can be identified with the Boolean algebra of $\mathbf V$-recognizable languages~\cite{Almeida:book}; in other words, the free pro-$\mathbf V$ monoid is the Stone dual of the Boolean algebra of $\mathbf V$-recognizable languages.  In recent years, a number of authors have made explicit use of duality theory to redevelop and expand the foundations of the profinite approach to studying varieties of languages in the sense of Eilenberg~\cite{Eilenberg76}; most closely related  to our work here is the work of Gehrke, Pin et al. \cite{GGP2008,GKP2016,GPR2016,Geh2016}, Boja\'{n}czyk \cite{Boj2015}, and Rhodes and the second-named author \cite[Chapter~8]{RS2009}.

Until the present work, there has been little attempt to use Stone duality to study the structure of relatively free profinite monoids.  For example, very basic structural properties of the free pro-aperiodic monoid were deduced in~\cite{RS01} using the closure of the pseudovariety of aperiodic monoids under certain expansions.  Deeper structural properties involving Green's relations on stabilizers and the equidivisibility property, which suggested that combinatorics on words can be extended in the limit to free pro-aperiodic monoids, were studied in~\cite{HRS2010} (see also~\cite{AC2009}).  Further structural properties were deduced in~\cite{Steinberg2010}.

The strongest result about the free pro-aperiodic monoid to date is McCammond's solution to the word problem for $\omega$-terms over aperiodic monoids (where we recall that $x^{\omega}$ is the idempotent power of $x$ in a finite semigroup).  McCammond gave an identity basis for the $\omega$-terms that are equal in all finite aperiodic monoids and provided normal forms to solve the word problem.  His proof is technically difficult and relies on his solution to the word problem for free Burnside semigroups of sufficiently high exponent~\cite{MAC91}, which is an even more challenging result.  For this reason, a number of authors have tried to come up with other approaches to this problem.  Almeida, Costa and Zeitoun~\cite{ACZ2015} gave a different proof of the correctness of McCammond's solution to the word problem that avoided using free Burnside semigroups~\cite{MAC91}.  Huschenbett and Kufleitner~\cite{HK2014} provided an alternate, more efficient algorithm for solving the word problem for $\omega$-terms over aperiodic monoids but they required McCammond's identity basis to prove the correctness of their algorithm. Their approach, like ours, makes use of the connection between first-order logic and aperiodic monoids.   Almeida, Costa and Zeitoun have since proved a number of structural results concerning free pro-aperiodic monoids, making use of McCammond normal forms and limiting techniques~\cite{ACZ09a,ACZ14}. 
 More recently, Almeida, Klíma and Kunc showed that the $\omega$-word problem is decidable for the free pro-objects in the concatenation hierarchy for aperiodic monoids \cite{AKK2016}.

In this paper, we use Stone duality and saturated models to obtain most of the known structural results about free pro-aperiodic monoids, as well as some new ones.  The crucial idea is that the identification of aperiodic-recognizable languages with first-order definable languages means that we can view elements of the free pro-aperiodic monoid as ultrafilters in the Boolean algebra of first-order definable languages, which in turn can be identified with the complete theories containing the theory of finite words.  By the compactness theorem of logic, this can be in turn identified with elementary equivalence classes of models of the theory of finite words, so-called pseudofinite words.  These models can be concatenated in a natural way, allowing us to recover the algebraic structure as well as the topological structure from this approach.

Models of the theory of finite words over an alphabet $A$ are examples of $A$-words (words over a discrete linear order with endpoints).  Each elementary equivalence class of $A$-words contains special models, called $\omega$-saturated models, that can realize, in a precise sense, its factorizations up to elementary equivalence.  This allows one to prove many of the structural results for  free pro-aperiodic monoids, like equidivisibility, in almost the same way as they are proved for free monoids.

One of our key technical results is that if we substitute $\omega$-saturated words into an $\omega$-saturated word, then the result is an $\omega$-saturated word.  This allows one to analyze the factors of the image of an element under a substitution in a very similar fashion to the case of free monoids.  Using this result, we can explicitly construct countable $\omega$-saturated models of $\omega$-terms (which turn out to be the same models considered in~\cite{HK2014}).  This allows us to give simple, transparent proofs of many of the results of~\cite{ACZ09a,ACZ2015}, such as the fact that factors of $\omega$-terms are $\omega$-terms, as well as providing a new correctness proof for the word problem algorithm of~\cite{HK2014}, avoiding entirely the machinery of McCammond~\cite{MAC91}.  Moreover, many of our results extend beyond just $\omega$-terms.  For example, in~\cite{ACZ2015} it is shown that factors of an $\omega$-term are well-quasi-ordered by divisibility and the finite factors form a regular language.  By showing that these properties behave well under substitution, we can both easily recover this result for $\omega$-terms and extend it beyond them.

In~\cite{HRS2010}, structural results about free pro-aperiodic monoids were used in order to give clean proofs of the decidability of aperiodic idempotent-pointlikes~\cite{Henckell04} and stable pairs~\cite{Henckell10} (i.e., one-vertex inevitable graphs in the sense of Almeida~\cite{Almeida2000}).  It is hoped that our approach will lead to simpler proofs of other algorithmic properties, as well as proofs of those needed to decide Krohn-Rhodes complexity~\cite{HRS2012}.
Further work would be to extend the theory to monadic second order logic to describe absolutely free profinite monoids.  This would, in particular, require a language rich enough to talk about $p$-adic numbers and free profinite groups!

\section*{Outline} 
In Section~\ref{sec:logic}, we recall the relevant notions from logic that we need, in particular, Ehrenfeucht-Fraiss\'e games. In Section~\ref{sec:proAP}, we recall how Stone duality can be used to relate the free pro-aperiodic monoid to logic. In Section~\ref{sec:subst} we show how substitutions and concatenations are interpreted from the logic point of view. In Section~\ref{sec:TAfaxiom} we make a brief excursion into the question of axiomatizability of the theory of finite words. In Section~\ref{sec:saturation} we recall the crucial notion of saturation from model theory, adapted to our specific context of words, and use topology to prove, as our main technical result, that saturated models behave well with respect to substitutions. In Section~\ref{sec:structure} we give a few basic applications of the theory developed thus far to the structure theory of the free pro-aperiodic monoid. In Section~\ref{sec:wpomega} we show how our results simplify the solution of the word problem for $\omega$-terms. In Section~\ref{sec:factorsomega} we use our results to analyze factor orders in our pro-aperiodic monoids. We prove in particular that several properties are stable under substitutions (and hence in particular under concatenation and $\omega$-power), including the properties that the $\J$-ordering on the set of factors of an element is a well-quasi-order and that the set of factors of an element is regular.

\section{Logic on words}\label{sec:logic}
\subsection{Basic definitions}
\emph{Regular languages} of finite words are those subsets of a finitely generated free monoid that are saturated with respect to a finite index monoid congruence. Regular languages can alternatively be characterized using finite automata or regular expressions, see, e.g., \cite{Str1994}.

Another important characterization of regular languages uses \emph{monadic second-order logic} \cite{Buc1960}. In this paper we only consider the \emph{first-order logic} fragment. We recall the necessary definitions. 
In the logical approach, a finite word $w = a_1 \cdots a_n$ over a finite alphabet $A$ is viewed as a \emph{relational structure}, in the sense of model theory, consisting of a binary relation $<$, capturing the ordering on positions in the word, and a unary relation $P_a$, for each $a\in A$, capturing the positions carrying the letter $a$. For example, the word $W = aba$ over the alphabet $\{a,b,c\}$ is based on the finite set $\{1, 2, 3\}$ with $<$ the order $1 < 2 < 3$, $P_a = \{1,3\}$, $P_b = \{2\}$, $P_c = \emptyset$.
Viewing words in this way, we can use logical sentences to define sets of words. For example, the first-order sentence $\exists x \exists y (P_a(x) \wedge P_b(y) \wedge x < y)$ defines the regular language $A^* a A^* b A^*$. Formally, we adopt the following definitions and notations, cf., e.g., Straubing's book \cite{Str1994} for more details.

Let $A$ be a finite alphabet. A \emph{word} over $A$, or \emph{$A$-word}, is a tuple $W = (|W|,<^W,(P_a^W)_{a \in A})$, where
\begin{itemize}[leftmargin=*]
\item $|W|$ is a set,
\item the relation $<^W$ is a discrete linear order with endpoints on $|W|$, i.e., there are a first and a last element, every element except the last has a unique immediate $<^W$-successor and every element except the first has a unique immediate $<^W$-predecessor,
\item $(P_a^W)_{a \in A}$ is a partition of $|W|$.
\end{itemize}
If $W$ is an $A$-word and $i \in |W|$, we write $W(i)$ for the unique letter $a$ such that $i \in P_a^W$.

An $A$-word is called \emph{finite} if $|W|$ is finite. We denote by $\emp$ the unique $A$-word with $|W| = \emptyset$.

An \emph{atomic formula} is an expression of the form $x < y$ or $P_a(x)$, where, here and in what follows, $x$ and $y$ denote first-order variables.
A \emph{first-order formula} is an expression built up from atomic formulas by inductively applying the connectives $\wedge$, $\vee$, $\neg$, $\to$, $\exists x$, and $\forall x$.
An occurrence of a variable in a formula $\phi$ is called \emph{bound} if it is under the scope of a quantification, the occurrence is called \emph{free} otherwise. If $\overline{x} = x_1,\dots,x_n$ is a tuple of first-order variables, the notation $\phi(\overline{x})$ indicates that all the variables that occur freely in $\phi$ lie in $\overline{x}$. A \emph{sentence} is a formula without free occurrences of variables.

If $W$ is a word, an \emph{assignment} of a tuple of first-order variables $\overline{x}$ in $W$ consists of a tuple $\overline{x}^W$ of the same length as $\overline{x}$, where each element of the tuple $\overline{x}^W$ is an element of $|W|$.

For a first-order formula $\phi(\overline{x})$, $W$ a word and $\overline{x}^W$ an assignment of $\overline{x}$ in $W$, we inductively define the relation $\models$ as follows:
\begin{itemize}[leftmargin=*]
\item $W, \overline{x}^W \models x_i < x_j$ iff $x_i^W <^W x_j^W$,
\item $W, \overline{x}^W \models P_a(x)$ iff $x^W \in P_a^W$,
\item $W, \overline{x}^W \models \neg \phi$ iff it is not the case that $W, \overline{x}^W \models \phi$,
\item $W, \overline{x}^W \models \phi \wedge \psi$ iff $W, \overline{x}^W \models \phi$ and $W, \overline{x}^W \models \psi$,
\item $W, \overline{x}^W \models \phi \vee \psi$ iff $W, \overline{x}^W \models \phi$ or $W, \overline{x}^W \models \psi$;
\item  $W, \overline{x}^W \models\exists y \phi$ iff there exists $y^W \in |W|$ with $W, \overline{x}^Wy^W \models \phi$,
\item  $W, \overline{x}^W \models\forall y \phi$ iff for all $y^W \in |W|$: $W, \overline{x}^Wy^W \models \phi$.
\end{itemize}
The \emph{language defined by an FO-sentence} $\phi$ is the set $L_\phi$ of finite words $W$ such that $W \models \phi$.



The first-order definable languages form a strict subclass of the regular languages. For example, the language $(AA)^*$, consisting of all words of even length is regular, but not first-order definable. Indeed, recall \cite{Sch1965,McNPap1971} that a language $L$ of finite $A$-words is first-order definable if and only if  the syntactic monoid of $L$ is finite and aperiodic, i.e., does not contain non-trivial subgroups. This fundamental result is the starting point for our perspective on the free pro-aperiodic monoid, cf. Section~\ref{sec:proAP} below.
%

\subsection{Quantifier depth and EF games}
A very useful classification of first-order formulas uses quantifier depth.
The \emph{quantifier depth} of a formula $\phi$ is the maximum nesting depth of quantifiers in $\phi$. 
If $U$ and $V$ are $A$-words, we write $U \equiv_k V$ if $U$ and $V$ satisfy exactly the same first-order sentences $\phi$ of quantifier depth at most $k$; in this case, we say that $U$ and $V$ are \emph{elementarily equivalent up to quantifier depth $k$} or simply \emph{$k$-equivalent}. We write $U \equiv V$ if $U \equiv_k V$ for all $k$, i.e., $U$ and $V$ satisfy exactly the same first-order sentences; in this case, we say that $U$ and $V$ are \emph{elementarily equivalent}, or simply \emph{equivalent}. Importantly, for each $k \geq 0$, there are only \emph{finitely} many $k$-equivalence classes, and each such class is definable by a first-order sentence. Indeed, each $k$-equivalence class is defined by a \emph{game-normal sentence of depth $k$}, and the set $\Theta_{0,k}$ of such sentences is finite \cite[Thm.~3.3.2]{Hod1993}, see also \cite[Sec.~VI.1]{Str1994}.

Below we will prove a first fundamental lemma about $k$-equivalence for words, which imports the well-known technique of Ehrenfeucht-Fraiss\'e games (EF games) into our context. We briefly summarize the idea of EF games, referring to \cite[Ch.~3]{Hod1993} for more information. Let $U$ and $V$ be relational structures. The \emph{$k$-round EF game on $(U,V)$} has two players, $\forall$ and $\exists$. In the $i^\mathrm{th}$ round of the game, $\forall$ begins by choosing a position in one of the two structures, denoted $u_i$ or $v_i$, and $\exists$ responds by choosing a position in the other structure, denoted $v_i$ or $u_i$. After $k$ rounds, the player $\exists$ \emph{wins} if 
$u_i \mapsto v_i$ is an isomorphism between the substructure $\overline{u}$ of $U$ and the substructure $\overline{v}$ of $V$.
The main result about EF games is that $U \equiv_{k} V$ if and only if  player $\exists$ has a \emph{winning strategy} in the $k$-round EF game played on $U$ and $V$.

We introduce a few useful notations. Let $U$ be an $A$-word and $i \in |U|$, we write $U\be{i}$ for the $A$-word obtained by restricting $U$ to the set of positions strictly less than $i$ (the \emph{ray} left of $i$), and similarly we define $U\af{i}$ (the \emph{ray} right of $i$). If $i$ and $j$ are positions in $U$, we write $U(i,j)$ for the \emph{open interval}, $U[i,j]$ for the \emph{closed interval} of $U$ between these positions, and \emph{half-open intervals} $U[i,j)$ and $U(i,j]$. Since the order on an $A$-word is discrete with endpoints, any (half-)open interval or ray can be written as a closed interval; we use this fact without mention in what follows. Recall that $U(i)$ denotes the unique letter $a$ such that $i \in P_a^U$. Finally, if $V$ is another $A$-word and $j \in |V|$, we say that the position $i$ in $U$ \emph{$k$-corresponds} to $j$ in $V$ if $U\be{i} \equiv_k V\be{j}$, $U(i) = V(j)$, and $U\af{i} \equiv_k V\af{j}$.

\begin{lemma}\label{lem:EFgames}
Let $U$ and $V$ be $A$-words. For all $k \geq 0$, $U \equiv_{k+1} V$ if and only if  for every $i \in |U|$, there exists $j \in |V|$ that $k$-corresponds to $i$, and for every $j \in |V|$, there exists $i \in |U|$ that $k$-corresponds to $j$.
\end{lemma}
\begin{proof}
For the left-to-right implication, suppose that $U \equiv_{k+1} V$ and $i \in |U|$. Let $j \in |V|$ be the element chosen by player $\exists$ according to her winning strategy in the $(k+1)$-round EF game on $(U,V)$ if player $\forall$ chooses $i$ in $U$ as his opening move. Then certainly $U(i) = V(j)$, and moreover player $\exists$ can win both $k$-round EF games on $(U\be{i},V\be{j})$ and on $(U\af{i},V\af{j})$ by following her winning strategy for the $(k+1)$-round EF game on $(U,V)$. Thus the position $j$ $k$-corresponds to $i$, as required. This suffices, by symmetry. 
For the right-to-left implication, suppose that every $i$ in $U$ $k$-corresponds to some $j$ in $V$ and vice versa. Then player $\exists$ can win the $(k+1)$-round EF game on $(U,V)$, as follows. If $\forall$ chooses a position $i$ in $U$ in the first round, $\exists$ responds by choosing a position $j$ in $V$ $k$-corresponding to $i$, and vice versa. In the remaining $k$ rounds, $\exists$ follows her strategy for either $(U\be{i},V\be{j})$ or $(U\af{i},V\af{j})$, or, if $\forall$ plays at $i$ in $U$ (respectively, $j$ in $V$), then  $\exists$ plays at $j$ in $V$ (respectively, $i$ in $U$). The resulting strategy is winning, because both strategies for $(U\be{i},V\be{j})$ and $(U\af{i},V\af{j})$ are winning.
\end{proof}

An immediate consequence of this lemma is a form of cancellability of first and last letters up to $\equiv_k$.

\begin{corollary}\label{c:cancel}
Let $a\in A$ and let $U,V$ be $A$-words.  Suppose that $k\geq 0$ and $aU\equiv_{k+1} aV$.  Then $U\equiv_{k} V$.  Dually, if $Ua\equiv_{k+1} Va$, then $U\equiv_{k}V$.  In particular, if $aU\equiv aV$, then $U\equiv V$ and, similarly, if $Ua\equiv Va$, then $U\equiv V$.
\end{corollary}
\begin{proof}
If $k=0$, there is nothing to prove. Assume that $k\geq 1$.  Since $\varepsilon$ is only $k$-equivalent to itself, the position of $aV$ that must $k$-correspond to the first position of $aU$ is the first position.  Thus $U\equiv_k V$. 
\end{proof}

\subsection{Pseudofinite words}
For a finite alphabet $A$, we denote by $\TA$ the (finitely axiomatized) \emph{theory of $A$-words}, that is, the set of first-order sentences deducible from axioms expressing that the order is a discrete linear order with endpoints, and that exactly one letter predicate holds at each position. We further let $\TAf$ denote the \emph{theory of finite $A$-words}, i.e., the set of first-order sentences that are true in all finite $A$-words.
A model of the theory $\TAf$ is called a \emph{pseudofinite $A$-word}.\footnote{This is an instance of the general model-theoretic use of the term `pseudofinite', cf., e.g., \cite{Vaa2003}.}
The theories $\TA$ and $\TAf$ do not coincide in general; in fact, they coincide only if the alphabet $A$ contains a single letter. In this case, both theories $\TA$ and $\TAf$ are  the theory of discrete linear orders with endpoints, with a unary predicate that is true everywhere. The following lemma about the one-letter case is well-known (cf., e.g., \cite[Thm.~IV.2.1]{Str1994}). It can also be proved easily using Lemma~\ref{lem:EFgames}.
\begin{lemma}\label{lem:oneletter}
For every $k \geq 0$, if $U$ and $V$ are $\{a\}$-words of cardinality at least $2^k-1$, then $U \equiv_k V$.\qed
\end{lemma}
In particular, if $A = \{a\}$, then any two infinite models of $\TAf$ are elementarily equivalent. The smallest infinite model of $\TAf$ is $\mathbb{N} + \mathbb{N}^\mathrm{op}$.

As soon as $A$ contains at least two letters, the situation is very different. In particular, there are $A$-words that are not pseudofinite.
\begin{example}\label{exa:TAnotTAfin}
Let $W$ be the word over the alphabet $\{a,b\}$ with underlying order $\mathbb{N} + \mathbb{N}^\op$, where $W(i) = a$ for all $i \in \mathbb{N}$ and $W(i) = b$ for all $i \in \mathbb{N}^\op$; visually, $W$ is the word
\[ aaaa \dots \;\; \dots bbbb.\]
The sentence
\[ \exists x P_a(x) \to \exists x (P_a(x) \wedge \forall y (y > x \to \neg P_a(y))\]
expressing `if there exists an $a$-position, then there exists a last such' is true in every finite $A$-word, and therefore lies in $\TAf$, but it fails to hold in $W$. Thus, $W$ is not pseudofinite.
\end{example}

A useful characterization of pseudofinite words is the following.
\begin{lemma}\label{lem:psfinchar}
An $A$-word $U$ is pseudofinite if and only if  for each $k \geq 0$ there exists a finite $A$-word $U_k$ such that $U \equiv_k U_k$.
\end{lemma}
\begin{proof}
If $U$ is pseudofinite, let $k \geq 0$, and let $\phi_k$ be the disjunction of those game-normal sentences of depth $k$ which define the $k$-equivalence class of a finite word. Then $\phi_k$ lies in $\TAf$, so $U$ satisfies it, and therefore $U$ also satisfies one of the disjuncts. The converse is clear.
\end{proof}
In Section~\ref{sec:TAfaxiom}, we discuss the axiomatizability of the theory of finite $A$-words. In particular, we show that the theory is not finitely axiomatizable, and we give an axiomatization of it using an axiom scheme similar to the one given by Doets \cite{Doe1987}.

\section{Pro-aperiodic monoids}\label{sec:proAP}
A \emph{profinite monoid} is an inverse limit of finite discrete monoids in the category of topological monoids (i.e., monoids whose underlying set is equipped with a topology in which the monoid operation is continuous). Equivalently, a profinite monoid is a topological monoid whose underlying space is a \emph{Boolean} or \emph{Stone space}, i.e., compact, Hausdorff and zero-dimensional. Profinite monoids inherit many properties of finite monoids. Of particular interest to us is the fact that any element $x$ in a profinite monoid has a unique idempotent, denoted $x^\omega$, in its orbit-closure $\overline{\{x^n \ | \ n \geq 1\}}$. A \emph{pro-aperiodic} monoid $M$ is a profinite monoid in which $x^{\omega} = x^{\omega}x$ for all $x \in M$. Equivalently, a pro-aperiodic monoid is an inverse limit of finite aperiodic monoids; here, finite monoids are equipped with the discrete topology, and the inverse limit is taken in the category of topological monoids.

The \emph{free pro-aperiodic monoid} generated by a finite set $A$ is a pro-aperiodic monoid $\FA(A)$ containing $A$ such that any function $f \colon A \to M$, with $M$ a finite aperiodic monoid, extends uniquely to a continuous homomorphism $\overline{f} \colon \FA(A) \to M$, where $M$ is given the discrete topology. The free pro-aperiodic monoid is unique up to topological isomorphism\footnote{Here and in what follows, we use the term `homeomorphism' to indicate an isomorphism in the category of topological spaces, and `topological isomorphism' for an isomorphism in the category of topological monoids.}, and the same extension property still holds if in the previous sentence $M$ is replaced by an arbitrary pro-aperiodic monoid.

We recall some basic definitions from the theory of monoids that are used throughout the paper. Let $M$ be a monoid and $u$, $v$ elements of $M$. Then $u \leq_{\J} v$ means that there exist $x$ and $y$ such that $u = xvy$, and in this case $v$ is called a \emph{factor} of $u$; $u \leq_{\L} v$ means that there exists $x$ such that $u = xv$, and in this case $v$ is called a \emph{suffix} of $u$; $u \leq_{\R} v$ means that there exists $y$ such that $u = vy$, and in this case $v$ is called a \emph{prefix} of $u$. Each of these relations is a quasi-order on $M$; the equivalence relations they induce are denoted $\J$, $\L$ and $\R$, e.g., $u \mathrel{\J} v$ means that $u \leq_\J v$ and $v \leq_\J u$. We will repeatedly use the easily proved facts that any homomorphism preserves Green's relations, and that, if $\alpha_1 \leq_{\L} \alpha_2$ and $\beta_1 \leq_{\R} \beta_2$, then $\alpha_1 u \beta_1 \leq_{\J} \alpha_2 u \beta_2$ for any $u$.
Recall that an element $e$ in a monoid $M$ is called \emph{idempotent} if $ee = e$. The set of idempotent elements in $M$ is denoted by $E(M)$. A \emph{(two-sided) ideal} $I$ in a monoid $M$ is a subset which is closed under multiplication by arbitrary elements from $M$ on both sides, i.e., if $a \in I$ and $u \in M$, then $ua \in I$ and $au \in I$. An ideal is \emph{idempotent} if $I = I^2$, i.e., for every $a \in I$, there exist $x, y \in I$ such that $a = xy$.  An ideal is \emph{prime} if its complement is a submonoid. 

We make use of \emph{Stone duality}, which, we briefly recall, is the dual equivalence between the categories of Boolean algebras and Boolean spaces that takes a Boolean space $X$ to its algebra $\K(X)$ of clopen sets, and a Boolean algebra $B$ to the set of ultrafilters $\St(B)$ of $B$, which is given a Boolean topology by declaring, for each $L \in B$, the set $\widehat{L} := \{x \in \St(B) \ | \ L \in x\}$ to be open. Stone's duality theorem \cite{Sto1936} says that the assignment $L \mapsto \widehat{L}$ is an isomorphism from $B$ to $\K(\St(B))$, and that moreover Boolean algebra homomorphisms $B_1 \to B_2$ are in natural bijection with continuous functions $\St(B_2) \to \St(B_1)$. 

The topological space underlying any profinite monoid is the Stone dual space of some Boolean algebra. For free objects, this Boolean algebra is particularly nice. The following result has a long history and holds in much more generality, but we do not go into this here; for more information, see, e.g., \cite[Sec.~4]{Geh2016} or~\cite[Sec.~3.6]{Almeida:book}.
\begin{theorem}\label{thm:produal}
Let $A$ be a finite alphabet. The Boolean space underlying the free pro-aperiodic monoid $\FA(A)$ is the Stone dual space of the Boolean algebra $\mathrm{Rec}_{\mathbf{A}}(A)$ of aperiodic-recognizable languages of finite $A$-words.
\end{theorem}

We already cited in the previous section the characterization theorem \cite{Sch1965,McNPap1971} that first-order languages are exactly the aperiodic-recognizable languages. We now connect this result with Theorem~\ref{thm:produal}. Recall that, for $T$ a set of sentences in first-order logic, the \emph{Lindenbaum-Tarski algebra} $\LT(T)$ of $T$ is the Boolean algebra of $T$-equivalence classes of first-order sentences. Here, two first-order sentences $\phi$ and $\psi$ are \emph{$T$-equivalent} if, in any model $W$ where all sentences in $T$ are true, $\phi$ and $\psi$ are either both true or both false.

\begin{proposition}\label{prop:lindalg}
The Lindenbaum-Tarski algebra $\LT(\TAf)$ of the theory $\TAf$ is isomorphic to the algebra $\mathrm{Rec}_{\mathbf{A}}(A)$ of aperiodic-recognizable languages of finite $A$-words.
\end{proposition}
\begin{proof}
Note that two first-order sentences $\phi$ and $\psi$ define the same language if and only if  $\phi$ and $\psi$ are $\TAf$-equivalent. Thus, the assignment $\phi \mapsto L_\phi$ is a well-defined injective function from $\LT(\TAf)$ to languages of finite $A$-words. It is clear from the definition of $L_\phi$ that this is a homomorphism. The image of this homomorphism consists exactly of the aperiodic-recognizable languages \cite{Sch1965,McNPap1971}.
\end{proof}
Theorem~\ref{thm:produal} and Proposition~\ref{prop:lindalg} together immediately entail that $\FA(A)$ is homeomorphic to the Stone dual space of $\LT(\TAf)$. Stone dual spaces of Lindenbaum-Tarski algebras are well understood in logic.
\begin{proposition}
Let $T$ be a set of first-order sentences. The Stone dual space of $\LT(T)$ is homeomorphic to the Boolean space $X$ whose points are elementary equivalence classes of models of $T$, in which the clopen sets are exactly the \emph{truth sets}
\[ \widehat{\phi} := \{x \in X \ | \ \phi \text{ is true in models in the class } x\},\]
for $\phi$ any first-order sentence.
\end{proposition}
\begin{proof}
To any model $W$ of $T$, associate the ultrafilter $\mathcal{U}_W$ of sentences that are true in $W$. By definition, $\mathcal{U}_W = \mathcal{U}_{W'}$ if and only if  $W$ and $W'$ are elementarily equivalent, so this is a well-defined injective function from $X$ to  $\St \LT(T)$. The completeness theorem for first-order logic entails that this map is also surjective. It is clear from the definitions that it is continuous and open, and therefore a homeomorphism.
\end{proof}
The following theorem is now an immediate application of the preceding three results.
\begin{theorem}\label{thm:start}
Let $A$ be a finite alphabet. The Stone space underlying $\FA(A)$ is homeomorphic to the Boolean space of elementary equivalence classes of pseudofinite $A$-words, with clopens the truth sets of first-order sentences.\qed
\end{theorem}
We will henceforth identify $\FA(A)$ with the space of classes of pseudofinite $A$-words. Since $\FA(A)$ is not just a topological space, but also carries a monoid multiplication, it is natural to ask how one can multiply two classes of pseudofinite $A$-words. This is done by simply concatenating $A$-words in each class and taking the class of the result, as we will see in the next section.

We will denote by $\AL$ the Stone dual space of $\LT(\TA)$. By Proposition~\ref{prop:lindalg}, elements of $\AL$ are elementary equivalence classes of $A$-words. In view of Theorem~\ref{thm:start}, $\FA(A)$ can be viewed as a subspace of $\AL$, also see Proposition~\ref{prop:FAinAL} below. It will often be useful to consider this larger space $\AL$ consisting of all classes of $A$-words, not just the pseudofinite ones.

The spaces $\AL$ and $\FA(A)$ can be usefully described as inverse limits of particular countable chains of finite discrete spaces, as follows. For each $k \geq 0$, let $\AL_k$ denote the set of $\equiv_k$-classes of $A$-words. As noted in Section~\ref{sec:logic}, the set $\AL_k$ is finite. Since $\equiv_{k+1}$ refines $\equiv_k$, and $\equiv$ refines each $\equiv_k$, we get a diagram
\begin{equation}\label{eq:ALchain}
\AL \onto \;\;\; \cdots \;\;\; \onto \AL_{k+1} \onto \AL_k \onto \;\;\; \cdots \;\;\; \onto \AL_0.
\end{equation}
The Stone dual of this diagram is the chain of subalgebra inclusions
\begin{equation}\label{eq:LTchain}
\LT(\TA)_0 \into \;\;\; \cdots \;\;\; \into \LT(\TA)_k \into \LT(\TA)_{k+1} \into \;\;\; \cdots \;\;\; \into \LT(\TA),
\end{equation}
where $\LT(\TA)_k$ is the finite Boolean algebra of $\TA$-equivalence classes of first-order sentences of depth at most $k$. Since $\LT(\TA)$ is clearly the direct limit (union) of the diagram in (\ref{eq:LTchain}), it follows from Stone duality that $\AL$ is the inverse limit as a topological space of the diagram in (\ref{eq:ALchain}), where each $\AL_k$ is regarded as a finite discrete space. For each $k$, let $\FAf(A)_k$ denote the subset of $\AL_k$ consisting of the $\equiv_k$-classes of \emph{finite} $A$-words. Again, $\FA(A)$ is the inverse limit in the diagram
\begin{equation}\label{eq:FAchain}
\FA(A) \onto \;\;\; \cdots \;\;\; \onto \FAf(A)_{k+1} \onto \FAf(A)_k \onto \;\;\; \cdots \;\;\; \onto \FAf(A)_0.
\end{equation}
Note that Example~\ref{exa:TAnotTAfin} shows that $\FAf(A)_2$ is strictly contained in $\AL_2$, since the word $W$ in that example fails to satisfy a depth $2$ sentence that holds in all (pseudo)finite words. We will show in the next section that these remarks still hold true when the spaces are endowed with a monoid structure.

Let us denote, for each $k \geq 0$, by $\pi_k$ the continuous projection map $\AL \onto \AL_k$. With a slight abuse of notation, if $u, u' \in \AL$, we write $u \equiv_k u'$ to mean $\pi_k(u) = \pi_k(u')$, i.e., $U \equiv_k U'$ for all $U$ in the class $u$ and $U'$ in the class $u'$.

We can now draw the following useful conclusions.
\begin{proposition}\label{prop:ALchain}
Let $A$, $B_1, \dots, B_n$ be finite alphabets.
\begin{enumerate}[leftmargin=*]
\item The collection $\{[U]_{\equiv_k} \ | \ k \geq 0, \, U \, \text{an $A$-word}\}$ is a basis of clopen sets for $\AL$.
\item A function $f \colon \Lambda(B_1) \times \cdots \times \Lambda(B_n) \to \AL$ is continuous if and only if  for each $k \geq 0$ there exists $m(k) \geq 0$ such that for any $v_i, v_i' \in \Lambda(B_i)$, if $v_i \equiv_{m(k)} v_i'$ then $f(v_1,\dots,v_n) \equiv_k f(v'_1,\dots,v'_n)$.
\end{enumerate}
\end{proposition}
\begin{proof}
(1) Clear from the fact that $\AL$ is the inverse limit of the finite discrete spaces $\AL_k$. 

(2) Using (1) and the definition of product topology, $f$ is continuous if and only if  for any $U$, $k$, the set $f^{-1}([U]_{\equiv_k})$ is a finite union of sets of the form $[V_1]_{\equiv_{m_1}} \times \cdots \times [V_n]_{\equiv_{m_n}}$. If $f$ is continuous,
choose $m(k)$ to be the maximum value of $m_i$ that occurs in such a finite union as $[U]_{\equiv_k}$ ranges over the (finitely many) elements of $\AL_k$. Conversely,  if the condition holds, then $f^{-1}([U]_{\equiv_k})$ is the union of $[V_1]_{\equiv_{m(k)}} \times \cdots \times [V_n]_{\equiv_{m(k)}}$ such that $f([V_1]_{\equiv}, \dots, [V_n]_{\equiv})$ lies in $[U]_{\equiv_k}$.
\end{proof}

\section{Substitutions and concatenation}\label{sec:subst}
Suppose that $V$ is a word over a finite alphabet $B$, and that for each $b \in B$, $U_b$ is a word over a finite alphabet $A$. We obtain a new word $V[b/U_b]$ over $A$ by substituting the word $U_b$ for each occurrence of the letter $b$ in $V$; see Figure~\ref{fig:substitution}. 
\begin{figure}[htp]
\begin{tikzpicture}[decoration={brace}]
\renewcommand{\labeldist}{.1}
\drawwordnr{(2,2)}{6}{$V$}

\drawwordnr{(0,0)}{10}{$V[b/U_b]$}

\drawwordn{(4,0)}{1.5}{}
\drawdecor{(5.4,-.2)}{(4.1,-.2)}{$U_b$}{}
\drawsubtri{(4.75,2)}{(4,0)}{(5.5,0)}{$b$}

\drawwordn{(8.5,0)}{1}{}
\drawdecor{(9.4,-.2)}{(8.6,-.2)}{$U_{b'}$}{}
\drawsubtri{(7.5,2)}{(8.5,0)}{(9.5,0)}{$b'$}

\node[above=.15] at (6.15,2) {$\dots$};
\node[below=.4] at (7,0) {$\dots$};
\end{tikzpicture}

\caption{Substituting $(U_b)_{b \in B}$ into $V$}
\label{fig:substitution}
\end{figure}
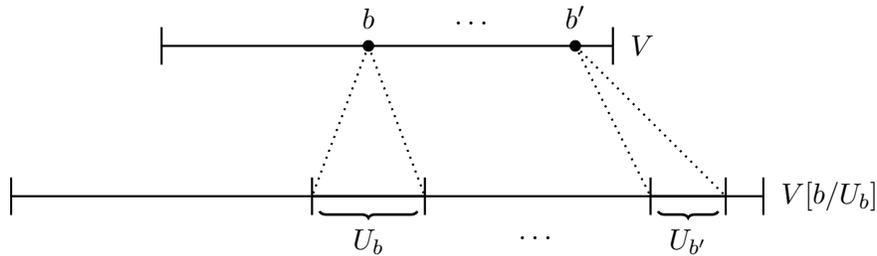

Formally, the \emph{substitution} of the $A$-words $(U_b)_{b \in B}$ into the $B$-word $V$ is the $A$-word $W = V[b/U_b]$ defined as follows. 
\begin{itemize}[leftmargin=*]
\item The underlying order of $W$ is the \emph{lexicographic order} on the disjoint union $|W| := \bigsqcup_{i \in |V|} |U_{V(i)}|$, i.e.,
\[ (i,j) <^W (i',j') \stackrel{\mathrm{def}}{\iff} i <^V i', \text{ or } i = i' \text{ and } j <^{U_{V(i)}} j'.\]
\item The letter at position $(i,j)$ in $W$ is the letter at position $j$ in $U_{V(i)}$.
\end{itemize}
There are two important special cases of substitution. If $U_0$ and $U_1$ are $A$-words, then the \emph{concatenation} $U_0 \cdot U_1$ of $U_0$ and $U_1$ is defined as the substitution of $(U_b)_{b \in \{0,1\}}$ into the $\{0,1\}$-word $01$. If $U$ is an $A$-word and $\lambda$ is a discrete linear order with endpoints, then the \emph{$\lambda$-power} $U^\lambda$ of $U$ is defined as the substitution of $U_b = U$ into the unique $\{b\}$-word with underlying order $\lambda$.

Importantly, the operation of substitution respects the equivalence relations $\equiv_k$. This is a well-known result in the model theory of labelled linear orderings, cf., e.g., \cite[Thm.~A.6.2]{Hod1993} and \cite[Sec.~13.6]{Ros1982}.
\begin{proposition}\label{prop:substinv-k}
Let $k \geq 0$. For any finite alphabets $A$ and $B$ and any $B$-indexed collections of $A$-words $(U_b)_{b \in B}$ and $(U_b')_{b \in B}$ such that $U_b \equiv_k U_b'$ for each $b \in B$, if $V$ and $V'$ are $A$-words such that $V \equiv_k V'$, then $V[b/U_b] \equiv_k V'[b/U_b']$.
\end{proposition}
\begin{proof}
By induction on $k$. The case $k = 0$ is trivial. Assume the statement is true for $k$.   Suppose that $U_b \equiv_{k+1} U_b'$ for each $b \in B$, and $V \equiv_{k+1} V'$. We use Lemma~\ref{lem:EFgames} to prove that $V[b/U_b] \equiv_{k+1} V'[b/U_b']$. Let $(i,j)$ be a position in $V[b/U_b]$. Since $V \equiv_{k+1} V'$, by Lemma~\ref{lem:EFgames}, pick $i' \in |V'|$ which $k$-corresponds to the position $i$ in $V$. In particular, $V(i) = V'(i')$; denote this letter by $b$. Since $U_{b} \equiv_{k+1} U'_b$, by Lemma~\ref{lem:EFgames}, pick $j'$ in $U'_b$ which $k$-corresponds to the position $j$ in $U_b$.
We finish the proof by showing that the position $(i',j')$ in $V'[b/U_b']$ $k$-corresponds to $(i,j)$ in $V[b/U_b]$. The letters at these positions are $U_b(j)$ and $U'_b(j')$, which are the same because $j'$ $k$-corresponds to $j$. Writing $Y$ for the subword $V[b/U_b]\be{(i,j)}$ and $Y'$ for the subword $V'[b/U'_b]\be{(i',j')}$, we use the induction hypothesis to show that $Y \equiv_k Y'$. Let $c$ be a new letter not in $A$ nor $B$. Define $U_c$ to be the $A$-word $U_b\be{j}$. Notice that $Y$ is the result of substituting the words $(U_b)_{b \in B \cup \{c\}}$ into the word $V\be{i}c$, obtained by appending a single letter $c$ to the end of $V\be{i}$. Similarly, $Y' = (V'\be{i'}c)[b/U'_b]$ where $U'_c := U_b\be{j'}$.
Now, since $V\be{i} \equiv_k V'\be{i'}$, we have $V\be{i}c \equiv_k V'\be{i'}c$ by the induction hypothesis applied in the special case of a concatenation. Also, $U_c \equiv_k U'_c$ by the choice of $j'$, and $U_b \equiv_k U'_b$ for each $b \in B$ since $(k+1)$-equivalence implies $k$-equivalence. By the induction hypothesis, $Y \equiv_k Y'$. By symmetry, we are done.
\end{proof}
\begin{corollary}\label{cor:substinv}
If $V \equiv V'$ and $U_b \equiv U_b'$ for each $b \in B$, then $V[b/U_b] \equiv V'[b/U_b']$.\qed
\end{corollary}

Corollary~\ref{cor:substinv} in particular implies that there is a well-defined binary operation of concatenation on the set $\AL$ of elementary equivalence classes of $A$-words. Recall from the previous section that this set is naturally a Stone space, being the dual space of the Lindenbaum-Tarski algebra of the theory $T_A$.
\begin{theorem}\label{thm:AL}
The Stone space $\AL$, equipped with the operation of concatenation up to elementary equivalence, is a pro-aperiodic monoid. The $\omega$-power of an element $[U]_{\equiv}$ is $[U^{\lambda}]_{\equiv}$, where $\lambda$ is any infinite discrete linear order with endpoints.
\end{theorem}
\begin{proof}
It is straight-forward to prove that concatenation is an associative operation with neutral element the class of an empty word. Moreover, concatenation is continuous by combining Proposition~\ref{prop:substinv-k} with Proposition~\ref{prop:ALchain}(2).

We prove the statement about the $\omega$-power. Let $u = [U]_{\equiv} \in \AL$ be arbitrary and let $\lambda$ be an infinite discrete linear order with endpoints. Let $v := [U^\lambda]_{\equiv}$. To prove that $v = u^\omega$, we need to show that $v$ is idempotent and lies in $\overline{\{u^n \ | \ n \geq 1\}}$. For idempotency, notice that $U^\lambda \cdot U^\lambda$ is isomorphic to $U^{\lambda + \lambda}$. Since the models $\lambda$ and $\lambda + \lambda$ are elementarily equivalent, it follows from Corollary~\ref{cor:substinv} that $U^\lambda \equiv U^{\lambda+\lambda} \cong U^\lambda \cdot U^\lambda$, so $v$ is idempotent. We now prove that $v$ lies in the orbit-closure of $u$. For every $k \geq 0$, by Lemma~\ref{lem:oneletter}, for $n_k := 2^k-1$ we have that $\lambda$ is $k$-equivalent to the finite linear order on $n_k$ elements. By Proposition~\ref{prop:substinv-k}, $v \equiv_k u^{n_k}$. It now follows from Proposition~\ref{prop:ALchain} that $v$ lies in the closure of $\{u^n \ | \ n \geq 1\}$.

Given the description of the operation $()^\omega$, it is easy to see that $\AL$ satisfies the equation $x^\omega = x^\omega x$; indeed, for any $A$-word $U$, the $A$-words $U^{\lambda}$ and $U^{\lambda}U$ are isomorphic.
\end{proof}

Recall from Theorem~\ref{thm:start} and the remarks following it that the free pro-aperiodic monoid, $\FA(A)$, is isomorphic to a subspace of $\AL$. Indeed, we can prove more.
\begin{lemma}\label{lem:productpseudofinite}
Let $U$ and $V$ be $A$-words. The concatenation $UV$ is pseudofinite if and only if  both $U$ and $V$ are pseudofinite.
\end{lemma}
\begin{proof}
We use the characterization of pseudofinite words in Lemma~\ref{lem:psfinchar}.\footnote{An alternative proof of this lemma can be given using the axiomatization of $\TAf$ in Proposition~\ref{prop:TAfax}.}
If both $U$ and $V$ are pseudofinite, then there exist sequences of finite $A$-words $(U_k)_{k \geq 0}$ and $(V_k)_{k \geq 0}$ such that $U \equiv_k U_k$ and $V \equiv_k V_k$ for all $k \geq 0$. By Proposition~\ref{prop:substinv-k}, $UV \equiv_k U_k V_k$ for all $k \geq 0$. Conversely, if $UV$ is pseudofinite, for any $k \geq 0$, pick $W_k$ such that $UV \equiv_{k+1} W_k$. Let $i$ denote the first position in the $V$-part of $UV$. Pick a position $j$ in $W_k$ which $k$-corresponds to $i$. Then $U \equiv_k W_k\be{j}$ and $V \equiv_k W_k({\geq}j)$. So both $U$ and $V$ are pseudofinite.
\end{proof}
\begin{proposition}\label{prop:FAinAL}
The set of classes of pseudofinite $A$-words is a topologically closed submonoid of $\AL$, which is the free pro-aperiodic monoid $\FA(A)$.  Moreover, it is the complement of a prime ideal of $\AL$.
\end{proposition}
\begin{proof}
\newcommand{\PF}{\mathrm{PF}(A)}
Let us denote by $\PF$ the set of classes of pseudofinite $A$-words in $\AL$. The set $\PF$ is closed because it can be written as the intersection $\bigcap_{\phi \in \TAf} \widehat{\phi}$. It is a submonoid (whose complement is a prime ideal) by Lemma~\ref{lem:productpseudofinite}. It now follows from Theorem~\ref{thm:AL} that $\PF$ is a pro-aperiodic monoid, and by Theorem~\ref{thm:start} its underlying space is homeomorphic to $\FA(A)$. The homeomorphism of Theorem~\ref{thm:start} sends any finite word $w$ in $\FA(A)$ to the elementary equivalence class of that finite word in $\PF$. Thus, the homeomorphism clearly preserves multiplication of finite words. Since the finite words are a dense subset of $\FA(A)$, it is an isomorphism for the monoid structure.
\end{proof}

\begin{convention}
In view of Proposition~\ref{prop:FAinAL}, we henceforth identify the free pro-aperiodic monoid $\FA(A)$ with the closed submonoid $\mathrm{PF}(A)$ of $\AL$ consisting of the elementary equivalence classes of pseudofinite $A$-words, and we will no longer make a notational distinction between the two. Moreover, throughout the paper except in Section~\ref{sec:wpomega}, we tacitly consider $A$-words up to isomorphism.
\end{convention}

Recall that in Section~\ref{sec:proAP} we defined $\AL_k$ to be the finite set of $\equiv_k$-equivalence classes of $A$-words. By Proposition~\ref{prop:substinv-k}, $\AL_k$ carries a monoid multiplication for each $k$, which is given in the same way as the multiplication on $\AL$. Hence, the chain in (\ref{eq:ALchain}) defined in the previous section, which exhibited the topological space $\AL$ as the inverse limit of finite sets $\AL_k$, now also exhibits the topological monoid $\AL$ as the inverse limit of finite monoids $\AL_k$. The same remarks apply to $\FA(A)$ and $\FAf(A)_k$.

Notice that, for every $k$, $\AL_k$ is an aperiodic monoid. This is true because $\AL_k$ is a finite quotient of the pro-aperiodic monoid $\AL$, but we can say something more precise. Indeed, for any $A$-word $W$, Proposition~\ref{prop:substinv-k} implies that $W^{2^k-1} \equiv_k W^{2^k}$, since $2^{k}-1 \equiv_k 2^k$ by Lemma~\ref{lem:oneletter}. Therefore, for any $w \in \AL_k$, we have $w^{2^k-1} = w^{2^k}$.

\begin{remark}\label{rem:blockproduct}
It is essentially implicit in~\cite{Str1994} that a formula $\phi$ has quantifier depth at most $k$ if and only if the syntactic monoid of the language of $\phi$ belongs to a strongly bracketed iterated block product of $k$ copies of the pseudovariety of semilattices.  In fact, using Lemma~\ref{lem:EFgames} and a well known description of the free object in a block product of locally finite pseudovarieties, due to Almeida (cf.~\cite{Almeida:book}), it is easy to see that $\FAf(A)_k$ is the relatively free monoid on $A$ in the strongly bracketed iterated block product of $k$ copies of the pseudovariety of semilattices.  We will show in future work that, more generally,  $\AL_k$ belongs to the strongly bracketed iterated block product of $k$ copies of the pseudovariety of semilattices.
\end{remark}

\begin{proposition}\label{prop:conthom}
Let $(U_b)_{b \in B}$ be a $B$-indexed collection of $A$-words. The function $f \colon \BL \to \AL$, which sends an element $[V]_{\equiv}$ of $\BL$ to $[V[b/U_b]]_\equiv$, is a well-defined continuous homomorphism.

Moreover, if each $U_b$ is pseudofinite, then $f$ restricts to a map $\FA(B) \to \FA(A)$, and any continuous homomorphism from $\FA(B)$ to $\FA(A)$ arises in this manner.
\end{proposition}
\begin{proof}
The function $f$ is well-defined by Corollary~\ref{cor:substinv}, it is continuous by Proposition~\ref{prop:substinv-k} and Proposition~\ref{prop:ALchain}.2, and it is a homomorphism because the $A$-words $(VV')[b/U_b]$ and $V[b/U_b] V'[b/U_b]$ are isomorphic.

For the moreover statement, notice that if each $U_b$ is pseudofinite and $V$ is pseudofinite, then $V[b/U_b]$ is pseudofinite. Indeed, using Lemma~\ref{lem:psfinchar}, for any $k \geq 0$, if $U_b'$ are finite $A$-words such that $U_b \equiv_k U_b'$ and $V'$ is a finite $B$-word such that $V \equiv_k V'$, then by Proposition~\ref{prop:substinv-k}, $V[b/U_b] \equiv_k V'[b/U_b']$, and the latter is a finite $A$-word.

If $g \colon \FA(B) \to \FA(A)$ is a continuous homomorphism, pick, for each $b \in B$, some $A$-word $U_b$ in the class $g(b)$. The substitution arising from this $B$-indexed collection of $A$-words is a continuous homomorphism that extends $g|B \colon B \to \FA(A)$, so it must coincide with $g$ by the uniqueness part of the universal property of $\FA(B)$.
\end{proof}
We call a continuous homomorphism $f \colon \BL \to \AL$ a \emph{substitution} if it arises as in Proposition~\ref{prop:conthom}. We suspect that there exist continuous homomorphisms $f \colon \BL \to \AL$ that are not substitutions, but we leave it to further work to exhibit a concrete example.

\section{Axiomatizing pseudofinite words}\label{sec:TAfaxiom}
We saw in Example~\ref{exa:TAnotTAfin} that if the alphabet $A$ contains at least two letters, then $\TA$ and $\TAf$ do not coincide. Using Proposition~\ref{prop:substinv-k}, we can prove more.

\begin{theorem}\label{t:finite.basis}
If the alphabet $A$ contains at least two letters, then $\TAf$ is not finitely axiomatizable.
\end{theorem}
\begin{proof}
Generalizing Example~\ref{exa:TAnotTAfin}, for each $k \geq 0$, consider the word $W_k$ defined as
\[ ab^{2^k}ab^{2^k}\dots \;\; \dots ab^{2^k-1}ab^{2^k-1}.\]
Then $W_k$ is not pseudofinite because it does not have a last occurrence of the factor $ab^{2^k}$.

{\bf Claim.} $W_k$ is $k$-equivalent to the finite word $(ab^{2^k-1})^{2^k-1}$.

{\it Proof of Claim.} Note first that $b^{2^k} \equiv_k b^{2^k-1}$ by Lemma~\ref{lem:oneletter}, so $ab^{2^k} \equiv_k ab^{2^k-1}$ by Proposition~\ref{prop:substinv-k}. Applying Proposition~\ref{prop:substinv-k} to the substitution of $U_c = ab^{2^k}$ and $U_d = U'_c = U'_d := ab^{2^k-1}$ into the $\{c,d\}$-word $c^\mathbb{N}d^{\mathbb{N}^\op}$, we get that $W_k$ is $k$-equivalent to $(ab^{2^k-1})^{\mathbb{N}+\mathbb{N}^\op}$. The latter, in turn, is $k$-equivalent to $(ab^{2^k-1})^{2^k-1}$ since $c^{\mathbb{N}+\mathbb{N^\op}}$ is $k$-equivalent to $c^{2^k-1}$ by Lemma~\ref{lem:oneletter}.\qed

Now, if $S$ is a finite set of sentences in $\TAf$, let $k$ be the maximum quantifier depth of sentences occurring in $S$. Then $W_k$ will satisfy all sentences in $S$, because it is $k$-equivalent to a finite word, but it is not a model of $\TAf$.  Thus $S$ does not axiomatize $\TAf$.
\end{proof}

Theorem~\ref{t:finite.basis} can be reformulated as saying that $\FA(A)$ is not an open subspace of $\AL$.  Indeed, since $\FA(A)$ is closed, it is open if and only if it is clopen, which occurs if and only if $\TAf$ is finitely axiomatizable.

There is a natural infinite axiomatization of $\TAf$, obtained by adding an induction scheme \cite[Thm.~3.1.1]{Doe1987}. We recall how this can be done. For every first-order formula $\phi(x)$ in one free variable, consider the first-order sentence
\[\mathrm{Last}_{\phi} \colon \exists x \phi(x) \to \exists x (\phi(x) \wedge \forall y (y > x \to \neg \phi(y))), \]
which expresses that if $\phi(x)$ is true for some position $x$, then there is a last position where $\phi(x)$ is true. Clearly, $\mathrm{Last}_{\phi}$ is true in every finite word for every first-order formula $\phi(x)$.
\begin{proposition}\label{prop:TAfax}
The theory of finite $A$-words, $\TAf$, is generated by adding to $\TA$ the sentences
$\mathrm{Last}_{\phi}$ for each first-order formula $\phi(x)$ in one free variable.
\end{proposition}
\begin{proof}
Let $U$ be an $A$-word which satisfies $\mathrm{Last}_{\phi}$ for each $\phi(x)$. Let $k \geq 0$ be arbitrary. Let $\phi(x)$ be a first-order formula expressing that `the subword $U\be{x}$ is $k$-equivalent to some finite word'; such a formula exists because $\equiv_k$ has finitely many classes, each of which is first-order definable, and one can relativize formulas with respect to a position, cf.~Lemma~\ref{lem:relativize}. Let $f$ denote the first position of the word $U$. Since $U\be{f}$ is the empty word, $\phi(f)$ holds. Since $U$ satisfies the scheme, pick the last position $i$ where $\phi(i)$ holds. Pick a finite word $V$ that is $k$-equivalent to $U\be{i}$. Then, by Proposition~\ref{prop:substinv-k}, $U({\leq}i) \equiv_k Va$, where $a := U(i)$. Since $i$ is the last position where $\phi(i)$ holds, we must have that $i$ is the last position in $U$, so $U$ is $k$-equivalent to the finite word $Va$. By Lemma~\ref{lem:psfinchar}, since $k$ was arbitrary, $U$ is pseudofinite.
\end{proof}

\section{Saturated words}\label{sec:saturation}
In this section we introduce the useful notions of types and saturation from model theory, and prove that substitution preserves saturation. We do not introduce the notions in their full model-theoretic generality, but apply them immediately to our context of words. We do this so that the statements of our results can be understood without knowledge of the model-theoretic background, although the proofs do rely on results from model theory. We will justify and explain our use of the model-theoretic terminology in the Appendix.

Let $A$ be a finite alphabet and let $U$ be an $A$-word. For any position $i \in |U|$, define the triple
\[t^U(i) := ([U\be{i}]_{\equiv},U(i),[U\af{i}]_{\equiv}),\]
an element of the Cartesian product $\AL \times A \times \AL$. We call $t^U(i)$ the \emph{type} of $i$ in $U$. We prove in Proposition~\ref{prop:onetypes} that this definition of type is equivalent to the usual definition of (complete $1$-)type in model theory. We refer to the product space $\AL \times A \times \AL$, where the middle component $A$ has the discrete topology, as the \emph{type space}.

We write $\RT(U)$ for the \emph{set of types realized in $U$}, i.e., $\RT(U)$ is the subset $\{t^U(i) \ | \ i \in |U|\}$ of the type space. If $V$ is an $A$-word elementarily equivalent to $U$ and $j \in |V|$, then we say the type $t^V(j)$ is \emph{consistent with $U$}. We write $\CT(U)$ for the \emph{set of types consistent with $U$}, i.e., $\CT(U)$ is the subset $\{t^V(j) \ | \ V \equiv U, j \in |V|\}$ of the type space.

The $A$-word $U$ is \emph{weakly saturated} if every type consistent with $U$ is realized in $U$, i.e., $\CT(U) = \RT(U)$. The $A$-word $U$ is \emph{$\omega$-saturated} if every closed interval $U[i,j]$ in $U$ is weakly saturated. Note that a closed interval in an $\omega$-saturated word is $\omega$-saturated. Also, any finite $A$-word is $\omega$-saturated. We prove in Proposition~\ref{prop:omegasatequiv} that this definition is equivalent to the usual definition of $\omega$-saturation in model theory. Following usual model theoretic terminology, we say that an $A$-word is \emph{countably saturated} if it is $\omega$-saturated and its underlying set is countable.

\begin{example}\label{exa:saturated}
Consider the case of a one-letter alphabet. All infinite $\{a\}$-words are elementarily equivalent (see Lemma~\ref{lem:oneletter}) and pseudofinite. Hence, $\FA(\{a\}) = \Lambda(\{a\})$ is topologically isomorphic to the topological monoid $\mathbb{N} \cup \{\omega\}$, i.e., the one-point compactification of $\mathbb{N}$ with the usual addition, where $\omega$ is an absorbing element.

The space of types of $\{a\}$-words is $\FA(\{a\}) \times \{a\} \times \FA(\{a\})$. Concretely, types of $\{a\}$-words are of one of the following four forms:
\begin{itemize}
\item $(a^n,a,a^m)$ for $n,m \in \mathbb{N}$;
\item $(a^n,a,a^\omega)$ for $n \in \mathbb{N}$;
\item $(a^\omega,a,a^m)$ for $m \in \mathbb{N}$;
\item $(a^\omega,a,a^\omega)$.
\end{itemize}

Consider the following infinite $\{a\}$-words:
\begin{enumerate}
\item $W_1 := a^{\mathbb{N} + \mathbb{N}^{\mathrm{op}}}$,
\item $W_2 := a^{\mathbb{N} + \mathbb{Z} + \mathbb{N}^{\mathrm{op}}}$,
\item $W_3 := a^{\mathbb{N} + \mathbb{Q} \times \mathbb{Z} + \mathbb{N}^{\mathrm{op}}}$, where $\mathbb{Q} \times \mathbb{Z}$ denotes the $\mathbb{Q}$-indexed lexicographic sum of copies of $\mathbb{Z}$.
\end{enumerate}
The word $W_1$ is not weakly saturated, because the elementarily equivalent word $W_2$ realizes the type $(a^\omega,a,a^\omega)$, which is not realized in $W_1$, that is, $(a^\omega,a,a^\omega) \in \CT(W_1) \setminus \RT(W_1)$. The word $W_2$ is weakly saturated, because it realizes all the types. However, $W_2$ is not $\omega$-saturated, because the ray to the left of $i$, where $i$ is any point in the summand $\mathbb{Z}$, is isomorphic to $W_1$, and not weakly saturated. Notice that any closed interval in the word in $W_3$ is either finite or isomorphic to $W_3$, using the well-known fact that any open interval in the order $\mathbb{Q}$ is isomorphic to $\mathbb{Q}$ (cf. e.g., \cite[p.~100]{Hod1993}). Since finite words and $W_3$ are weakly saturated, the word $W_3$ is in fact $\omega$-saturated, and even countably saturated.
\end{example}

One advantage of $\omega$-saturated words is that they realize any finite factorization of their elementary equivalence class.

\begin{lemma}\label{lem:satfact}
Let $W$ be an $\omega$-saturated $A$-word and suppose that $W\equiv W_1\cdots W_n$ with $W_1,\ldots, W_n$ non-empty $A$-words.  Then we can find positions $i_1<i_2<\cdots<i_{n-1}$ in $|W|$ such that $W\be{i_1}\equiv W_1$, $W({\geq}i_{n-1})\equiv W_n$ and $W[i_j,i_{j+1})\equiv W_{j+1}$ for $1\leq j\leq n-2$.
\end{lemma}
\begin{proof}
We proceed by induction on $n$ (for all $\omega$-saturated $A$-words).  If $n=1$, then there is nothing to prove.  Assume it is true for $n-1$ and $W\equiv W_1\cdots W_n$.  Put $U=W_1\cdots W_{n-1}$.  Since $W$ is weakly saturated, we can find $i_{n-1}\in |W|$ such that $W\be{i_{n-1}}\equiv U$ and $W({\geq}i_{n-1})\equiv W_n$.  Then since $W\be{i_{n-1}}$ is $\omega$-saturated, by induction we can find $i_1<i_2<\cdots<i_{n-2}$ with $i_{n-2}<i_{n-1}$ and $W\be{i_1} \equiv W_1$ and $W[i_j,i_{j+1})\equiv W_{j+1}$ for $1\leq j\leq n-2$.  This completes the proof.
\end{proof}

\begin{lemma}\label{lem:satclosure}
Let $U$ be an $A$-word.
The set $\CT(U)$ is the topological closure of $\RT(U)$ in the type space.
In particular, $U$ is weakly saturated if and only if  $\RT(U)$ is closed.
\end{lemma}
\begin{proof}
Notice that the function $\AL \times A \times \AL \to \AL$ which sends $(u,a,v)$ to $uav$ is continuous, since multiplication is continuous. Therefore, $\CT(U)$, which is the inverse image of the point $[U]_{\equiv}$ under this map, is closed. To see that $\CT(U)$ is the closure of $\RT(U)$, let $t \in \CT(U)$, and let $\hat \phi \times \{a\} \times \hat \psi$ be an arbitrary basic neighbourhood of $t$. By definition of $\CT(U)$, pick $V \equiv U$ and $j \in |V|$ such that $t = t^V(j)$. Let $k$ be the maximum of the quantifier depths of $\phi$ and $\psi$. Since in particular $U \equiv_{k+1} V$, by Lemma~\ref{lem:EFgames} pick $i \in |U|$ which $k$-corresponds to the position $j$ in $V$. Then $U(i) = a$, and, since $V\be{j} \models \phi$, we have $U\be{i} \models \phi$, and similarly $U\af{i} \models \psi$. Thus, $t' = t^U(i)$ is a point of $\RT(U)$ that lies in $\hat \phi \times \{a\} \times \hat \psi$.
\end{proof}
It is well-known in model theory that any elementary equivalence class contains an $\omega$-saturated model, which typically has an uncountable underlying set; see Proposition~\ref{prop:satexist} in the Appendix.

Our main theoretical result about $\omega$-saturated words is that they are \emph{stable under substitutions} (Theorem~\ref{thm:omegasubst}). We first prove this for weakly saturated words, and then deduce it for $\omega$-saturated words.

For the proof of our next theorem, some more notation will be useful. We denote by $\alpha$ the continuous two-sided action of $\AL$ on the type space, defined for $t = (x,a,y) \in \AL \times A \times \AL$ and $u, v\in \AL$ by $\alpha(u,v,t) := (ux,a,yv)$. If $S \subseteq \AL \times \AL$ and $T \subseteq \AL \times A \times \AL$, we write $S \circ_\alpha T$ for the set $\{\alpha(u,v,t) \ | \ (u,v) \in S, t \in T\}$. For each $a \in A$, we denote by $\iota_a$ the continuous inclusion of $\AL \times \AL$ into the type space which sends $(u,v)$ to $(u,a,v)$.
\begin{theorem}\label{thm:satsubst}
If $V$ is a weakly saturated $B$-word and $(U_b)_{b \in B}$ are $A$-words that are weakly saturated, then $V[b/U_b]$ is weakly saturated.
\end{theorem}
\begin{proof}
Fix $A$-words $(U_b)_{b \in B}$ and denote by $f \colon \BL \to \AL$ the continuous homomorphism which sends $[V]_{\equiv}$ to $[V[b/U_b]]_{\equiv}$ (Proposition~\ref{prop:conthom}).
Note that, if $i \in |V|$, $b = V(i)$, and $j \in |U_b|$,
then, writing $t^{V}(i) = (u_0,b,u_1)$,
the definition of substitution gives that $t^{V[b/U_b]}((i,j)) = \alpha(f(u_0),f(u_1),t^{U_b}(j))$. This observation can be written as the following equality:
\begin{equation}\label{eq:RTsub}
\RT(V[b/U_b]) = \bigcup_{b \in B} (f \times f)[\iota_b^{-1}(\RT(V))] \circ_\alpha \RT(U_b).
\end{equation}
Now suppose that the words $V$ and $U_b$ are weakly saturated. By Lemma~\ref{lem:satclosure} the sets $\RT(V)$ and $\RT(U_b)$ are closed. Since $\iota_b$, $f$ and $\alpha$ are continuous, and therefore closed, maps and $B$ is finite, the equality (\ref{eq:RTsub}) implies that $\RT(V[b/U_b])$ is closed, so $V[b/U_b]$ is weakly saturated by Lemma~\ref{lem:satclosure}.
\end{proof}

Since concatenations and powers are particular cases of substitutions, we immediately obtain the following corollary.
\begin{corollary}\label{cor:satconpower-weak}
The concatenation of two weakly saturated words is weakly saturated. The power of a weakly saturated word by a weakly saturated linear order is weakly saturated.\qed
\end{corollary}

We will now deduce an analogous result for $\omega$-saturation. The following basic, but important, lemma analyzes intervals in substituted words. By a \emph{prefix} of an $A$-word $U$ we mean an $A$-word of the form $U\be{i}$ for some $i \in |U|$, and by a \emph{suffix} of $U$ we mean an $A$-word of the form $U\af{i}$ for some $i \in |U|$.
\begin{lemma}\label{lem:intervalinsubst}
Let $V$ be a non-empty $B$-word and $(U_b)_{b \in B}$ a $B$-indexed collection of $A$-words. For any interval $W = V[b/U_b]([k_1,k_2])$ in the $A$-word $V[b/U_b]$:
\begin{enumerate}
\item either there exists $b \in B$ such that $W$ is an interval in $U_b$,
\item or there exist $b_1, b_2 \in B$, a suffix $X$ of $U_{b_1}$, a prefix $Y$ of $U_{b_2}$, and a $B$-word $Z$ such that $b_1 \cdot Z \cdot b_2$ is an interval in $V$, and $W = X \cdot Z[b/U_b] \cdot Y$.
\end{enumerate}
\end{lemma}
\begin{figure}[htp]
\begin{tikzpicture}[decoration={brace}]

\renewcommand{\labeldist}{.3}
\drawwordnr{(2,2)}{6}{$V$}

\drawwordnr{(0,0)}{10}{$V[b/U_b]$}

\drawwordn{(4,0)}{1.5}{}
\drawdecor{(5.4,-.5)}{(4.1,-.5)}{$U_b$}{}
\drawsubtri{(4.75,2)}{(4,0)}{(5.5,0)}{$b$}

\drawwordo{(4.3,0)}{.7}{}{dashed}
\draw[thick] (4.3,0) -- (5,0);
\node[below] at (4.65,0) {\small $W$};
\end{tikzpicture}

\vspace{5mm}


\begin{tikzpicture}[decoration={brace}]

\renewcommand{\labeldist}{0.1}

\drawwordnr{(2,2)}{6}{$V$}

\drawwordnr{(0,0)}{10}{$V[b/U_b]$}

\drawwordn{(3,0)}{3}{$Z[b/U_b]$}
\drawwordn{(4.15,2)}{.7}{$Z$}
\draw[thick,dotted] (4.15,2)--(3,0);
\draw[thick,dotted] (4.85,2)--(6,0);

\drawsubtri{(4,2)}{(1.1,0)}{(3,0)}{$b_1$}
\drawwordn{(1.1,0)}{2}{\small $U_{b_1}$}

\drawsubtri{(5,2)}{(6,0)}{(8,0)}{$b_2$}
\drawwordn{(6,0)}{2}{\small $U_{b_2}$}

\node[below] at (2.75,0) {\small $X$};
\node[below] at (6.25,0) {\small $Y$};

\drawwordo{(2.5,0)}{4}{}{dashed}
\draw[thick] (2.5,0) -- (6.5,0);

\drawdecor{(6.3,-.5)}{(2.7,-.5)}{$W = XZ[b/U_b]Y$}{}

\end{tikzpicture}
\caption{Intervals in a substitution (Lemma~\ref{lem:intervalinsubst})}
\label{fig:interval}
\end{figure}
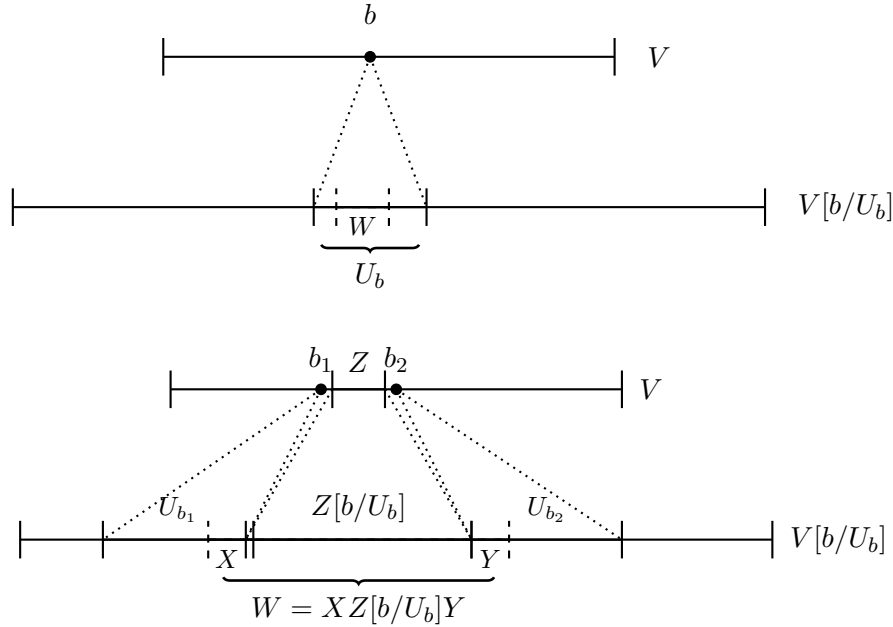
\begin{proof}
For $r = 1, 2$, we have $k_r = (i_r,j_r)$ for some $i_r \in |V|$ and $j_r \in |U_{V(i_r)}|$. There are two cases.

(1) $i_1 = i_2$. In this case, letting $b := V(i_1) = V(i_2)$, we have $W = U_b[j_1,j_2]$.

(2) $i_1 < i_2$. Let $b_r := V(i_r)$, $X$ the suffix of $V_{b_1}$ beginning at $j_1$, $Y$ the prefix of $V_{b_2}$ ending at $j_2$, and $Z$ the interval $V(i_1,i_2)$. Then $b_1 \cdot Z \cdot b_2 = V[i_1,i_2]$, and $W = X \cdot Z[b/U_b] \cdot Y$.
\end{proof}

\begin{theorem}\label{thm:omegasubst}
If $V$ is an $\omega$-saturated $B$-word and $(U_b)_{b \in B}$ is a $B$-indexed collection of $A$-words that are $\omega$-saturated, then $V[b/U_b]$ is $\omega$-saturated.
\end{theorem}
\begin{proof}
Let $W = V[b/U_b]([k_1,k_2])$ be a closed interval in $V[b/U_b]$. We prove that $W$ is weakly saturated. By Lemma~\ref{lem:intervalinsubst}, there are two cases.

(1) If $W$ is an interval in $U_b$ for some $b$, then $W$ is weakly saturated because $U_b$ is $\omega$-saturated by assumption.

(2) Suppose that $W = X \cdot Z[b/U_b] \cdot Y$ for some suffix $X$ of $U_{b_1}$, prefix $Y$ of $U_{b_2}$ and interval $Z$ in $V$.  Then $Z$ is weakly saturated because $V$ is $\omega$-saturated, so $Z[b/U_b]$ is weakly saturated by Theorem~\ref{thm:satsubst}, and $X$ and $Y$ are weakly saturated since the $U_b$ are $\omega$-saturated by assumption. By Corollary~\ref{cor:satconpower-weak}, the concatenation of weakly saturated words is weakly saturated.
\end{proof}

As before, Theorem~\ref{thm:omegasubst} has the following immediate consequence.
\begin{corollary}\label{cor:satconpower-omega}
The concatenation of two $\omega$-saturated words is $\omega$-saturated. The power of an $\omega$-saturated word by an $\omega$-saturated linear order is $\omega$-saturated.\qed
\end{corollary}

As a first application of Theorem~\ref{thm:omegasubst}, we will now analyze what factors of the result of a substitution can look like in $\AL$ (and hence, by Lemma~\ref{lem:productpseudofinite}, in $\FA(A)$).

\begin{proposition}\label{prop:findfact}
Let $W$ be an $\omega$-saturated $A$-word and let $v\in \AL\setminus \{\emp\}$.
\begin{enumerate}
\item  If $[W]_{\equiv}\leq_{\J}v$, then there exist positions $i<j$ in $|W|$ with $[W[i,j]]_{\equiv}=v$.
\item If $[W]_{\equiv}\leq_{\R}v$, then there exists a position $i\in |W|$ with $[W\be{i})]_{\equiv}=v$.
\item If $[W]_{\equiv}\leq_{\L}v$, then there exists a position $i\in |W|$ with $[W\af{i}]_{\equiv}=v$.
\end{enumerate}
\end{proposition}
\begin{proof}
Choose an $A$-word $V$ with $[V]_{\equiv} =v$.
To prove (1), there are $A$-words $X,Y$ with $W\equiv XVY$.  Note that $X$ or $Y$ could be empty. Now apply Lemma~\ref{lem:satfact}.  The proofs for (2) and (3) are similar.
\end{proof}

\begin{theorem}\label{thm:substfactors}
Let $f \colon \BL \to \AL$ be a substitution. Let $v \in \BL$ and $w \in \AL$.
\begin{enumerate}[leftmargin=*]
\item $f(v) \leq_{\J} w$ if and only if one of the following holds:
\begin{enumerate}
\item $w = \emp$, or
\item there exists $b \in B$ such that $v \leq_{\J} b$ and $f(b) \leq_{\J} w$, or
\item there exist $b_1,b_2 \in B$, $x, y \in \AL$ and $z \in \BL$ such that $v \leq_\J b_1zb_2$,  $f(b_1) \leq_{\L} x$, $f(b_2) \leq_{\R} y$, and $w = x f(z) y$.
\end{enumerate}
\item $f(v) \leq_{\R} w$ if and only if  $w=\emp$ or there exist $x \in \AL$, $b \in B$, and $z \in \BL$ such that $f(b) \leq_{\R} x$, $v \leq_{\R} zb$, and $w = f(z)x$.
\item $f(v) \leq_{\L} w$ if and only if  $w=\emp$ or there exist $y \in \AL$, $c \in B$, and $z \in \BL$ such that $f(c) \leq_{\L} y$, $v \leq_{\L} cz$, and $w = yf(z)$.
\end{enumerate}
\end{theorem}
\begin{proof}
The case $w = \emp$ is trivial, so we may henceforth assume $w \neq \emp$.
By Proposition~\ref{prop:satexist}, for each $b \in B$ pick an $\omega$-saturated $A$-word $U_b$ in the class $f(b)$, pick an $\omega$-saturated $B$-word $V$ in the class $v$, and also pick an $A$-word $W$ in the class $w$. Since $f$ is a substitution, $f(v) = [V[b/U_b]]_{\equiv}$. By Theorem~\ref{thm:omegasubst}, the $A$-word $V[b/U_b]$ is $\omega$-saturated.

(1) Suppose $f(v) \leq_{\J} w$.  Since $V[b/U_b]$ is $\omega$-saturated, there are  positions $k_1$ and $k_2$ in $V[b/U_b]$ such that $V[b/U_b]([k_1,k_2]) \equiv W$ by Proposition~\ref{prop:findfact}. Consider the interval $W' := V[b/U_b]([k_1,k_2])$. By Lemma~\ref{lem:intervalinsubst}, there are two cases.

(a) The $A$-word $W'$ is an interval in some $U_b$. In this case, $f(b) \leq_{\J} w$.

(b) The $A$-word $W' = X \cdot Z[b/U_b] \cdot Y$ for some suffix $X$ of $U_{b_1}$, prefix $Y$ of $U_{b_2}$ and $b_1 \cdot Z \cdot b_2$ an interval in $V$. Setting $z := [Z]_{\equiv}$, $x := [X]_{\equiv}$ and $y := [Y]_{\equiv}$ gives the desired result.

Conversely, if (a) holds, then $f(v) \leq_{\J} f(b) \leq_{\J} w$, since $f$ preserves the $\J$-ordering. Suppose (b) holds. Then $f(v) \leq_{\J} f(b_1) f(z) f(b_2) \leq_{\J} x f(z) y = w$.

(2) Suppose $f(v) \leq_{\R} w$. Since $V[b/U_b]$ is $\omega$-saturated,  pick (by Proposition~\ref{prop:findfact}) a position $k$ such that $W \equiv V[b/U_b]\be{k}$. Pick $i \in |V|$ and $j \in |U_{V(i)}|$ such that $k = (i,j)$. Set $b := V(i)$. Then $W \equiv Z[b/U_b] \cdot X$, and $V = ZbU$, where $Z := V\be{i}$, $U := V\af{i}$, and $X := U_{b}\be{j}$. Thus, setting $z := [Z]_{\equiv}$ and $x := [X]_{\equiv}$ gives the desired result.

Conversely, if $x$, $b$ and $z$ are given, then $f(v) \leq_{\R} f(z) f(b) \leq_{\R} f(z) x = w$.

(3) Follows from (2) by symmetry.
\end{proof}
We may usefully summarize Theorem~\ref{thm:substfactors} by recalling some more notation. For any $u \in \AL$, write ${\uparrow}_{\J}u := \{w \in \AL \ | \ w \geq_{\J} u\}$ for the set of factors of $u$, and similarly ${\uparrow}_{\R}u$ for the set of prefixes of $u$, and ${\uparrow}_{\L}u$ for the set of suffixes of $u$. Also, for any subset $L$ of $\AL$ and $a,b \in A$, write $a^{-1}L := \{u \in \AL \ | \ au \in L\}$ and $Lb^{-1} := \{u \in \AL \ | \ ub \in L\}$, and $a^{-1}Lb^{-1} := \{u \in \AL \ | \ aub \in L\}$. Finally, if $u \in \AL$, write $\Cont(u) := \{a \in A \ | \ u \leq_{\J} a\}$ for the \emph{content} of $u$. The following restates Theorem~\ref{thm:substfactors} in this notation.
\begin{corollary}\label{cor:substfactors}
Let $f \colon \BL \to \AL$ be a substitution. For any $v \in \BL$, we have:
\begin{align*}
{\uparrow}_{\J} f(v) &= \{\emp\} \cup \bigcup_{b \in \Cont(v)} \left({\uparrow}_{\J} f(b)\right) \cup \bigcup_{b_1,b_2 \in B} \left[ {\uparrow}_{\L} f(b_1) \cdot f\left(b_1^{-1}\left({\uparrow}_{\J}v\right)b_2^{-1}\right) \cdot {\uparrow}_{\R} f(b_2) \right], \\
{\uparrow}_{\R} f(v) &= \{\emp\}\cup \bigcup_{b \in B} f\left(\left({\uparrow}_{\R}v\right)b^{-1}\right) \cdot \left({\uparrow}_{\R}f(b)\right), \\
{\uparrow}_{\L} f(v) &= \{\emp\}\cup \bigcup_{c \in B} \left({\uparrow}_{\L}f(c)\right) \cdot f\left(c^{-1}\left({\uparrow}_{\L}v\right)\right).
\end{align*}
\end{corollary}

Theorem~\ref{thm:substfactors} can be viewed as extension of~\cite[Lemma~8.2]{AV2006}, where the special case that $v$ is a finite word is handled in the context of the free profinite monoid.

\section{Structure theory of $\AL$ and $\FA(A)$}\label{sec:structure}
In this section, we illustrate how the model-theoretic methods developed in the previous sections can be used to derive known results about $\FA(A)$ in a simple way, as well as new ones. Moreover, since these results do not depend on the words involved being pseudofinite, we are actually able to show the same results hold in the larger pro-aperiodic monoid $\AL$.

The following properties of $\AL$ will be used in Section~\ref{sec:factorsomega}. They are well-known for $\FA(A)$, and indeed items (1) and (2) in Lemma~\ref{lem:stabilityetc} are known to hold in any compact monoid. Items (3) and (4) in Lemma~\ref{lem:stabilityetc} say that the $\L$- and $\R$-orders on $\AL$ and $\FA(A)$ are \emph{unambiguous}. The latter result was first proved in \cite{RS01}.
\begin{lemma}\label{lem:stabilityetc}
For any $w, x, y \in \AL$:
\begin{enumerate}
\item if $w \leq_{\J} xw$, then $w \leq_{\L} xw$,
\item if $w \leq_{\J} wy$, then $w \leq_{\R} wy$,
\item if $w \leq_{\L} x$ and $w \leq_{\L} y$, then $x \leq_{\L} y$ or $y \leq_{\L} x$,
\item if $w \leq_{\R} x$ and $w \leq_{\R} y$, then $x \leq_{\R} y$ or $y \leq_{\R} x$,
\item if $w \leq_{\L} x$ and $w \leq_{\L} y$, then $x \leq_{\J} y$ implies $x \leq_{\L} y$,
\item if $w \leq_{\R} x$ and $w \leq_{\R} y$, then $x \leq_{\J} y$ implies $x \leq_{\R} y$.
\end{enumerate}
The same holds true for $\FA(A)$.
\end{lemma}
\begin{proof}
(1) If $w \leq_{\J} xw$, pick $\alpha, \beta \in \AL$ such that $w = \alpha x w \beta$. Repeatedly substituting $\alpha x w \beta$ for $w$, we see that $w = (\alpha x)^n w \beta^n$ for every $n$. Therefore, $w = (\alpha x)^\omega w \beta^\omega$, by continuity. Since $(\alpha x)^\omega$ is idempotent, we get from this that $w = (\alpha x)^\omega w$. Now, using aperiodicity,
\[w =  (\alpha x)^\omega w = (\alpha x)^\omega \alpha x w,\]
so that indeed $w \leq_{\L} xw$. (2) is dual to (1).

(3) Let $W$ be a weakly saturated $A$-word in the class $w$. Since $w \leq_{\L} x$ and $w \leq_{\L} y$, pick $i, j$ such that $[W[{\geq} i]]_{\equiv} = x$ and $[W[{\geq} j]]_{\equiv} = y$. If $i \leq j$, then letting $\alpha := [W[i,j)]_{\equiv}$ gives $x = \alpha y$, so $x \leq_{\L} y$, and, similarly, if $i > j$, then $y \leq_{\L} x$. (4) is dual to (3).

(5) By (3), since $w \leq_{\L} x$ and $w \leq_{\L} y$, we either have $x \leq_{\L} y$ or $y \leq_{\L} x$. If $y \leq_{\L} x$, pick $\alpha$ such that $y = \alpha x$. Then $x \leq_{\J} y = \alpha x$, so (1) gives $x \leq_{\L} \alpha x = y$, as required. (6) is dual to (5).

The result for $\FA(A)$ follows from the result for $\AL$ because the former is the complement of a prime ideal in the latter.
\end{proof}
The following result on factors of a specific form will be useful in Section~\ref{sec:factorsomega}. It could be deduced from Theorem~\ref{thm:substfactors}, but we prefer to give a direct proof.
\begin{lemma}\label{lem:bzc}
Let $z$, $z'$ be $A$-words and $b,c \in A$. If $bzc \leq_{\J} bz'c$, then one of the following four statements is true:
\begin{enumerate}
\item $z = z'$,
\item $z \leq_{\R} z'c$,
\item $z \leq_{\L} bz'$,
\item $z \leq_{\J} bz'c$.
\end{enumerate}
\end{lemma}
\begin{proof}
Let $Z$ be an $\omega$-saturated word in the class $z$ and let $Z'$ be an $A$-word in the class $z'$. Then $V := b \cdot Z \cdot c$ is an $\omega$-saturated word in the class $bzc$, by Corollary~\ref{cor:satconpower-omega}. Let $\bot$ denote the first position in $V$ and $\top$ denote the last position in $V$, so that $V(\bot) = b$, $V(\bot,\top) = Z$ and $V(\top) = c$. Since $bzc \leq_{\J} bz'c$, pick $i, j \in |V|$ such that $V[i,j] \equiv bZ'c$ (using Proposition~\ref{prop:findfact}). In particular, $V(i) = b$, $V(j) = c$, and $V(i,j) \equiv Z'$. There are four cases:

(1) $i = \bot$ and $j = \top$. In this case, $Z = V(\bot,\top) = V(i,j) \equiv Z'$, so $z = z'$.

(2) $i = \bot$ and $j \neq \top$.  Let $U$ be the subword $V(j,\top)$ of $V$. Now $Z = V(\bot,\top) = V(i,j) c U \equiv Z' c U$. So $z = z'cu$, and hence $z \leq_{\R} z'c$.

(3) $i \neq \bot$ and $j = \top$. By a proof analogous to (2), $z \leq_{\L} bz'$.

(4) $i \neq \bot$ and $j \neq \top$. Writing $U_1 := V(\bot,i)$ and $U_2 := V(j,\top)$, we have $Z = V(\bot,\top) = U_1 b V(i,j) c U_2 \equiv U_1 b Z' c U_2$. Hence, $z = u_1bz'cu_2$, where $u_k := [U_k]_{\equiv}$ ($k = 1,2$), and thus $z \leq_{\J}bz'c$.
\end{proof}

Recall that a monoid $M$ is called \emph{equidivisible} if for all $u,v,u',v' \in M$, if $uv = u'v'$ then there exists $x \in M$ such that either $ux = u'$ and $xv' = v$, or $u'x = u$ and $xv = v'$. The following result for $\FA(A)$ was implicitly proved in~\cite{HRS2010} (but not explicitly formulated), using limiting arguments, and explicitly proved in~\cite{AC2009}.
\begin{proposition}\label{prop:equidiv}
The monoids $\AL$ and $\FA(A)$ are equidivisible.
\end{proposition}
\begin{figure}[htp]
\begin{tikzpicture}[decoration=brace]
\renewcommand{\labeldist}{.1}
\drawwordn{(0,2)}{3}{$U$}
\drawwordn{(3,2)}{7}{$V$}

\drawwordnr{(0,0)}{10}{$W$}

\drawwordn{(0,-2)}{6}{$U'$}
\drawwordn{(6,-2)}{4}{$V'$}

\newcommand{\ipos}{3}
\draw[thick,dashed] (\ipos,1.75) to (\ipos,.25);
\draw[thick] (\ipos,-.15) to (\ipos,.15);
\draw[thick,->] (\ipos,.6) to (\ipos,.25);
\node[below] at (\ipos,-.1) {$i$};
\node (i) at (\ipos,0) {};

\newcommand{\jpos}{6}
\draw[thick,dashed] (\jpos,-1.75) to (\jpos,-.25);
\draw[thick] (\jpos,-.15) to (\jpos,.15);
\draw[thick,->] (\jpos,-.6) to (\jpos,-.25);
\node[above] at (\jpos,.1) {$j$};
\node (j) at (\jpos,0) {};

\drawdecor{(5.85,-.4)}{(3.15,-.4)}{$X$}{}
\end{tikzpicture}
\caption{Equidivisibility (Proof of Proposition~\ref{prop:equidiv})}
\label{fig:equidiv}
\end{figure}
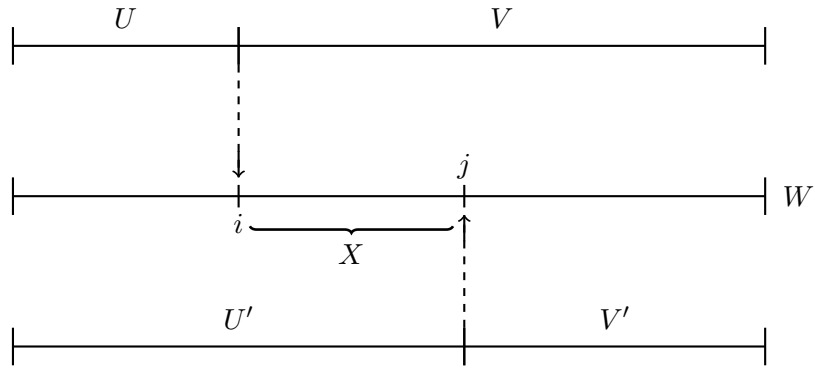
\begin{proof}
Let $U$, $V$, $U'$ and $V'$ be $A$-words such that $UV \equiv U'V'$, and choose a weakly saturated $A$-word $W$ in this class. Since $W$ is weakly saturated, pick positions $i, j \in |W|$ such that $W\be{i} \equiv U$, $W({\geq}i) \equiv V$, $W\be{j} \equiv U'$ and $W({\geq}j) \equiv V'$. Suppose without loss that $i \leq j$. Let $X$ be the $A$-word $W[i,j)$. Then $UX \equiv W\be{j} \equiv U'$ and $XV' \equiv W({\geq}i) \equiv V$. Moreover, if the $A$-words $U$ and $U'$ are pseudofinite, then $X$ is pseudofinite by Lemma~\ref{lem:productpseudofinite}.
\end{proof}

We deduce a few consequences of equidivisibility.

Recall that the \emph{right stabilizer} of an element $u$ in a monoid $M$ is the submonoid $\Stab_M(u) := \{s \in M \ | \ us = u\}$. The next result for $\FA(A)$ is from~\cite{HRS2010}. It was used there to establish the decidability of membership in a number of semidirect products of the form $\mathbf{V} \ast \mathbf{A}$. 
\begin{corollary}\label{cor:stabtotal}
For any $u \in \AL$, the quasi-order $\leq_{\L}$ on the right stabilizer $\Stab_{\AL}(u)$ is total. If $u \in \FA(A)$, the quasi-order $\leq_{\L}$ on the right stabilizer $\Stab_{\FA(A)}(u)$ is also total.
\end{corollary}
\begin{proof}
Suppose that $s_1, s_2 \in \Stab_{\AL}(u)$. Then $us_1 = us_2$. Proposition~\ref{prop:equidiv} yields $x \in \Stab_{\AL}(u)$ such that $xs_2 = s_1$ or $xs_1 = s_2$. If $u \in \FA(A)$, then also $x \in \FA(A)$, by Lemma~\ref{lem:productpseudofinite}.
\end{proof}
By symmetry, the dual result to Corollary~\ref{cor:stabtotal} holds for left stablizers and the quasi-order $\leq_{\R}$.

The following result was proved by the second-named author for $\FA(A)$~\cite{Steinberg2010}.
\begin{lemma}\label{lem:idpotprime}
Any idempotent ideal in $\AL$ or $\FA(A)$ is prime. In particular, if $e \in E(\AL)$, then the  ideal $\AL{e}\AL$ generated by $e$ is prime.
\end{lemma}
\begin{proof}
Let $I$ be an idempotent ideal and let $ab \in I$. Pick $x, y \in I$ such that $ab = xy$. By Proposition~\ref{prop:equidiv}, we have either $a \leq_{\R} x$ or $b \leq_{\L} y$. Therefore, one of $a$ and $b$ lies in $I$. For the `in particular' statement, notice that the ideal generated by an idempotent element is idempotent.
\end{proof}
In the following two propositions, we use the well-known fact that in a finite aperiodic semigroup, and hence in a pro-aperiodic semigroup, $\mathcal{H}$-classes are trivial, where $\mathcal{H}$ is the intersection of the equivalence relations $\mathcal{L}$ and $\mathcal{R}$.  See~\cite{Almeida:book,RS2009}.  First we extend a result from~\cite{Steinberg2010} to $\AL$.

\begin{proposition}
Suppose that $u\in \AL$ has finite order, i.e., generates a finite subsemigroup.  Then $u$ is an idempotent.
\end{proposition}
\begin{proof}
Since the submonoid generated by $u$ is closed, it must be aperiodic and so $u^n=u^{n+1}$ for some $n>0$.  Therefore, $u^n$ is idempotent.  But then since $u^n$ generates a prime ideal, we must have that $u$ and $u^n$ are $\mathcal J$-equivalent.  By Lemma~\ref{lem:stabilityetc}, we have that $u$ and $u^n$ are $\mathcal H$-equivalent and hence equal.  Thus $u$ is an idempotent.
\end{proof}

\begin{proposition}
If $u \in \AL$ and $e \in E(\AL)$ are such that $ue = eu$, then $ue \in \{u,e\}$.
\end{proposition}
\begin{proof}
Suppose that $ue = eu$. By Proposition~\ref{prop:equidiv}, there exists $x$ such that either (1) $ex = u = xe$ or (2) $ux = e = xu$.

(1) If $ex = u = xe$, then $eu = eex = ex = u$.

(2) Suppose that $ux = e = xu$.
Notice that, for every $n \geq 1$, $x^n$ and $e$ commute.
By Lemma~\ref{lem:stabilityetc}.3, for every $n \geq 1$, since $x^ne = ex^n$ is a common lower bound of $e$ and $x^n$ in $\leq_{\L}$, the elements $e$ and $x^n$ are comparable in $\leq_{\L}$. We further distinguish two cases.

(2a) There exists $n \geq 1$ such that $x^n \leq_{\L} e$. Then $x^n$ lies in the ideal generated by $e$, which is prime by Lemma~\ref{lem:idpotprime}. Hence, $x$ lies in the ideal generated by $e$, i.e., $x \leq_{\J} e$. Since $e = ux$, Lemma~\ref{lem:stabilityetc}.1 implies that $x \leq_{\L} e$, and similarly, $x \leq_{\R} e$. Since we already have $e \leq_{\R} x$ and $e \leq_{\L} x$, we conclude that $x$ is $\mathcal{H}$-equivalent to $e$. Since $\AL$ is pro-aperiodic, we obtain $x = e$, and hence $ue = ux = e$.

(2b) For every $n \geq 1$, $e <_{\L} x^n$. Since $\leq_{\L}$ is a closed relation in any compact monoid \cite[Proposition~3.1.9]{RS2009}, we get $e \leq_{\L} x^\omega$. Pick $v$ such that $e = vx^\omega$. Since $x^\omega$ is idempotent, we obtain $ex^\omega = v(x^\omega)^2 = vx^\omega = e$. Therefore, using that $x^\omega x = x$ in the pro-aperiodic monoid $\AL$, we have
\[ eu = ex^\omega u = ex^\omega x u = e x^\omega e = e^2  = e,\] as required.
\end{proof}

\section{The word problem for aperiodic $\omega$-terms}\label{sec:wpomega}
In this section, we use our techniques to give an improved proof of the decidability of the word problem for $\omega$-terms in $\FA(A)$.  The original proof of decidability was due to McCammond~\cite{Cam2001}, who introduced normal forms based on an inductively defined notion of rank.   McCammond shows, using an infinite basis of natural identities satisfied by $\omega$-terms, that each $\omega$-term can be placed into normal form.  The construction of the normal forms is quite technical and it is not clear how efficient it is to find the normal form of  a term  from the viewpoint of time complexity.  By far,  the most difficult part of McCammond's paper is the proof that distinct normal forms represent distinct elements of $\FA(A)$.  The separation of  normal forms makes use of his solution to the word problem for free Burnside semigroups of sufficiently large exponent~\cite{MAC91}, which inspired the definition of the normal forms in the first place. Of particular importance is that free Burnside semigroups of sufficiently large exponent have a system of co-finite ideals with empty intersection and so the quotients by these ideals can be used to distinguish the normal forms for $\omega$-terms.  Recently, Almeida, Costa and Zeitoun~\cite{ACZ2015} have provided a new proof that distinct McCammond normal forms represent distinct $\omega$-terms in $\FA(A)$ by introducing star-free languages associated to normal forms whose closures can be used to separate them.

Huschenbett and Kufleitner~\cite{HK2014} have developed a new algorithm to solve the word problem for $\omega$-terms in $\FA(A)$ using model-theoretic ideas.  From the point of view of our paper, they assign an $A$-word to each $\omega$-term whose elementary equivalence class represents the corresponding element of $\FA(A)$.  They then prove that two such interpretations of $\omega$-terms are elementarily equivalent to each other if and only if they are isomorphic.  They use work of Bloom and Esik~\cite{BloEsi2005} to prove that isomorphism can be decided in exponential time in the size of the expression as an $\omega$-term; in the conference presentation of this result they announced that this can be improved to polynomial time using recent work of Lohrey and Mathissen~\cite{LM2013}. The catch is that Huschenbett and Kufleitner rely on McCammond's work to prove correctness of the algorithm: the proof of \cite[Proposition~5.2]{HK2014} goes through the non-trivial direction of McCammond's normal forms~\cite{Cam2001} (see also~\cite{ACZ2015}).

Our work allows us to give a direct proof of the correctness of the Huschenbett and Kufleitner algorithm, circumventing McCammond's results entirely.  Our proof uses the same interpretation of $\omega$-terms as in \cite{HK2014}, and provides a new justification for that interpretation, in the following sense. We show (Proposition~\ref{prop:Ut}) that the interpretation of $\omega$-terms in \cite{HK2014} picks out a countably saturated $A$-word in the elementary equivalence class represented by an $\omega$-term.  That two $\omega$-terms, thus interpreted, represent the same element of $\FA(A)$ if and only if they are isomorphic is then an immediate consequence of the  standard model-theoretic fact that elementarily equivalent countably saturated models are isomorphic.

We begin with a few definitions. Let $A$ be a finite alphabet. An \emph{$\omega$-term} over $A$ is a term built up from finite words by using concatenation and $\omega$-power. If $M$ is a profinite monoid containing the alphabet $A$, then any $\omega$-term $t$ has a natural interpretation $\sem{t}_M$ in $M$. In the case $M = \FA(A)$, we will now inductively define, for any $\omega$-term $t$, a particular $A$-word $U_t$ in the class $\sem{t}_{\FA(A)}$. Let $\rho$ denote the linear order $\mathbb{N} + \mathbb{Q} \times \mathbb{Z} + \mathbb{N}^{\mathrm{op}}$, which is countably saturated (cf. Example~\ref{exa:saturated}).

\begin{itemize}
\item If $t$ is a term representing a finite word, let $U_t$ be that finite word.
\item If $t = t_1 \cdot t_2$, let $U_t$ be the $A$-word $U_{t_1} \cdot U_{t_2}$.
\item If $t = s^\omega$, let $U_t$ be the $A$-word $(U_s)^{\rho}$.
\end{itemize}

\begin{proposition}\label{prop:Ut}
For any $\omega$-term $t$, the $A$-word $U_t$ is a countably saturated $A$-word in the elementary equivalence class $\sem{t}_{\FA(A)}$.
\end{proposition}
\begin{proof}
Finite words are countably saturated. Concatenations and $\rho$-powers of countably saturated $A$-words are clearly countable, and, by Corollary~\ref{cor:satconpower-omega}, $\omega$-saturated. An easy induction, using Theorem~\ref{thm:AL} for the step involving $\omega$-power, shows that $U_t$ lies in $\sem{t}_{\FA(A)}$.
\end{proof}

\begin{theorem}
For any $\omega$-terms $t_1$, $t_2$, the following are equivalent:
\begin{enumerate}
\item $\sem{t_1}_{\FA(A)} = \sem{t_2}_{\FA(A)}$,
\item $U_{t_1}$ is isomorphic to $U_{t_2}$.
\end{enumerate}
\end{theorem}
\begin{proof}
(2) $\Rightarrow$ (1) is clear, since isomorphic $A$-words are elementarily equivalent. (1) $\Rightarrow$ (2). By Proposition~\ref{prop:Ut}, both $U_{t_1}$ and $U_{t_2}$ are countably saturated models in the same elementary equivalence class. By the uniqueness of countably saturated models (cf., e.g., \cite[Thm.~2.3.9]{CK}), $U_{t_1}$ and $U_{t_2}$ are isomorphic.
\end{proof}
In order to decide the word problem for $\omega$-terms in $\FA(A)$, one can now proceed as in \cite{HK2014} and use a decidability procedure for isomorphism of regular words (cf.~\cite{BloEsi2005}  or \cite{LM2013}) to decide isomorphism of the countably saturated $A$-words interpreting the $\omega$-terms.

\section{Factors}\label{sec:factorsomega}
In this section, we exploit our model-theoretic view of $\FA(A)$ (and $\AL$) to analyze the $\J$-orderings on sets of \emph{factors} of elements in $\AL$. We recover several of the structural results on factors of $\omega$-terms that were obtained in \cite{ACZ09a,ACZ14} using McCammond's normal forms as special cases of more general results. Again, our proofs avoid such normal forms altogether; instead, we use Theorem~\ref{thm:substfactors} on factors of a substitution. The following result was first proved by Almeida, Costa, and Zeitoun \cite[Theorem~7.4]{ACZ14}.

\begin{theorem}\label{thm:omegafact}
Prefixes, suffixes and factors of elements in $\FA(A)$ that are interpretations of $\omega$-terms are again interpretations of $\omega$-terms.
\end{theorem}
\begin{proof}
We prove the statement for prefixes. The statement for suffixes then follows by symmetry, and the statement for factors, in turn, then follows because any factor is a suffix of a prefix. 
We write $\sem{t}$ as shorthand for $\sem{t}_{\FA(A)}$. We will prove by induction on the complexity of an $\omega$-term $t$ that, for any $w \in \FA(A)\setminus \{\emp\}$ such that $\sem{t} \leq_{\R} w$, we have $w = \sem{s}$ for some $\omega$-term $s$.
%
Prefixes of a finite word are finite words.
If $t = t_0 \cdot t_1$ for some $\omega$-terms $t_0$ and $t_1$, and $\sem{t} \leq_{\R} w$, we apply Theorem~\ref{thm:substfactors} in the special case of a concatenation. If $\sem{t_0} \leq_{\R} w$, then we are done immediately by induction. Otherwise, we have $w = \sem{t_0} \cdot v$ for some $v$ with $\sem{t_1} \leq_{\R} v$. By induction, pick an $\omega$-term $s'$ such that $\sem{s'} = v$. Then for the $\omega$-term $s := t_0s'$, we have $\sem{s} = w$. If $t = r^\omega$ for some $\omega$-term $r$, and $\sem{t} \leq_{\R} w$, we apply Theorem~\ref{thm:substfactors} in the special case of an $\omega$-power. There exist $z \in \FA(1) = \mathbb{N} \cup \{\omega\}$ and $x \in \FA(A)$ with $\sem{r} \leq_{\R} x$ such that $w = \sem{r}^zx$. By the induction hypothesis, pick an $\omega$-term $s'$ such that $\sem{s'} = x$. The $\omega$-term $s := r^z s'$ gives $\sem{s} = w$.
\end{proof}
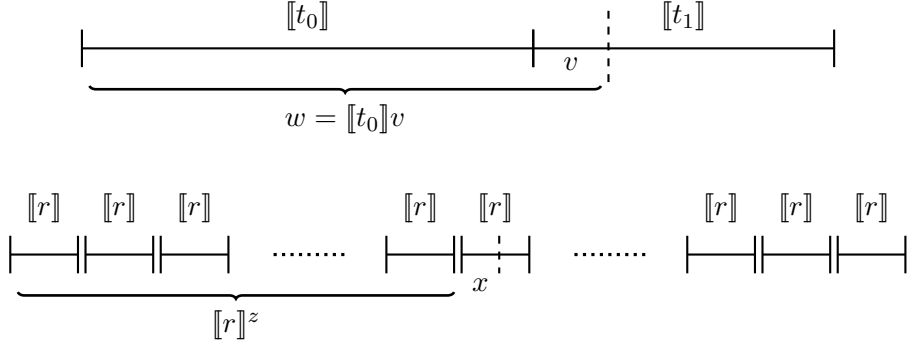
\begin{figure}[htp]
\begin{tikzpicture}[decoration=brace]
\renewcommand{\labeldist}{.1}
\drawwordn{(0,0)}{6}{$\sem{t_0}$}
\drawwordn{(6,0)}{4}{$\sem{t_1}$}

\draw[thick, dashed] (7,.5) to (7,-.5);

\node at (6.5,-.2) {$v$};

\drawdecor{(6.9,-.5)}{(.1,-.5)}{$w = \sem{t_0}v$}{}
\end{tikzpicture}

\vspace{5mm}

\begin{tikzpicture}[decoration=brace]
\renewcommand{\labeldist}{.3}

\drawwordn{(0,0)}{.9}{$\sem{r}$}
\drawwordn{(1,0)}{.9}{$\sem{r}$}
\drawwordn{(2,0)}{.9}{$\sem{r}$}
\draw[very thick,dotted] (3.5,0) -- (4.5,0);
\drawwordn{(5,0)}{.9}{$\sem{r}$}
\drawwordn{(6,0)}{.9}{$\sem{r}$}
\draw[very thick,dotted] (7.5,0) -- (8.5,0);
\drawwordn{(9,0)}{.9}{$\sem{r}$}
\drawwordn{(10,0)}{.9}{$\sem{r}$}
\drawwordn{(11,0)}{.9}{$\sem{r}$}

\draw[thick, dashed] (6.5,.25) to (6.5,-.25);
\node at (6.25,-.4) {$x$};
\drawdecor{(5.9,-.5)}{(.1,-.5)}{$\sem{r}^z$}{}
\end{tikzpicture}
\caption{Factors of interpretations of $\omega$-terms (Theorem~\ref{thm:omegafact})}
\end{figure}

We now focus on special properties of the $\J$-, $\R$- and $\L$-orders on the sets of factors, prefixes and suffixes of elements of $\AL$.

It is shown in~\cite[Theorem~5.1]{ACZ14} that an interpretation of an $\omega$-term has only finitely many regular $\J$-classes above it. Here, recall that a $\J$-class is \emph{regular} if it contains an idempotent element. For an element $w$ in a monoid $M$, we write \[\Reg(w) := \{ J_e \ | \ w \leq_{\J} e, e\ \text{idempotent}\}\] for the set of regular $\J$-classes above $w$ where, as usual, $J_e$ denotes the $\J$-class of $e$.  The original proof relies on McCammond normal forms, whereas our proof is fairly elementary from Theorem~\ref{thm:substfactors}.  Our proof also provides a simple algorithm to compute the regular $\J$-classes above the interpretation of an $\omega$-term from the $\omega$-term.

\begin{theorem}\label{thm:finJ}
Let $u, w_1, w_2 \in \AL$. Then
\begin{enumerate}
\item $\Reg(w_1 w_2) = \Reg(w_1) \cup \Reg(w_2)$.
\item $\Reg(u^\omega) = \Reg(u) \cup \{J_{u^\omega}\}$.
\end{enumerate}
In particular, the interpretation of any $\omega$-term has finitely many regular $\J$-classes above it.
\end{theorem}
\begin{proof}
(1) The right-to-left inclusion is clear since $w_i \geq_{\J} w_1w_2$ for $i = \{1,2\}$. Conversely, if $w_1w_2 \leq_{\J} e$ for some idempotent $e$, then $w_1 \leq_{\J} e$ or $w_2 \leq_{\J} e$ by Lemma~\ref{lem:idpotprime}.

(2) The right-to-left inclusion is clear since $u \geq_{\J} u^\omega$. Conversely, suppose that $e$ is an idempotent such that $u^\omega \leq_{\J} e$. By Theorem~\ref{thm:substfactors}, there exist $z \in \mathbb{N} \cup \{\omega\}$, $x \geq_{\L}u$ and $y \geq_{\R} u$ such that $e = xu^zy$. If $z \in \mathbb{N}$, then $e \geq_{\J} u^z$ implies that $e \geq_{\J} u$, by Lemma~\ref{lem:idpotprime}, so $J_e \in \Reg(u)$. If $z = \omega$, then $e\leq_{\J} u^{\omega}\leq_{\J} e$ and so $J_e=J_{u^{\omega}}$.

The final statement follows from (1) and (2) by induction on the complexity of the $\omega$-term defining $w$, noticing that, for a finite word $w$, $\Reg(w) = \{J_{\emp}\}$.
\end{proof}

We recall the notion of \emph{well-quasi-order} (wqo), which is a quasi-order satisfying the equivalent properties in the following proposition.
\begin{proposition}\label{prop:wqo}
Let $(Q,\preceq)$ be  a quasi-order. The following are equivalent:
\begin{enumerate}
\item there are no infinite antichains or infinite descending chains in $Q$;
\item for every infinite sequence $(q_i)_{i \in \omega}$ in $Q$, there exist $i < j$ such that $q_i \preceq q_j$;
\item every infinite sequence in $Q$ contains a non-decreasing subsequence.
\end{enumerate}
\end{proposition}
\begin{proof}
See, e.g., \cite[{\S}10.3]{Ros1982}.
\end{proof}

We will call an element $u \in \AL$ \emph{well-factor-ordered} (wfo) if the reverse $\J$-order, $\geq_{\J}$, is a well-quasi-order on the set ${\uparrow}_{\J}u$ of factors of $u$. Similarly, we call $u \in \AL$ \emph{well-prefix-ordered} (wpo) if $\geq_{\R}$ is a wqo on ${\uparrow}_{\R}u$ and \emph{well-suffix-ordered} (wso) if $\geq_{\L}$ is a wqo on ${\uparrow}_{\L}u$.
Recall from Lemma~\ref{lem:productpseudofinite} that $\FA(A)$ is upward closed in $\AL$ with respect to the $\J$-, $\R$- and $\L$-orders, so that the following results apply immediately to $\FA(A)$, as well.

Almeida showed~\cite[Theorem~2.6]{Almeida05} that $\FA(A)\setminus A^*$ contains maximal elements with respect to the $\J$-ordering and that they correspond in a sense that can be made precise to uniformly recurrent words. Moreover, it is shown in~\cite[Proposition~3.2]{ACZ14} that every element of $\FA(A)\setminus A^*$ is $\J$-below a maximal element. Notice that $\omega$ is the unique maximal element of $\FA(1)\setminus \mathbb N$.

\begin{proposition}
Let $w$ be a maximal element of $\AL\setminus A^*$ in the $\J$-order (e.g., a $\J$-maximal element of $\FA(A)\setminus A^*$).  Then $w$ is well-factor-ordered.
\end{proposition}
\begin{proof}
Let $w_1,w_2,\ldots$ be an infinite sequence of factors of $w$.  We must show that $w_i\geq_{\J} w_j$ for some $i<j$.  If  $w_i \mathrel{\J} w_j$ for some $i<j$, then we are done so we may assume that the elements of the sequence are in distinct $\J$-classes.  By maximality of $w$, at most one element of the sequence does not belong to $A^*$ and so, by passing to a subsequence, we may assume that $w_1,w_2,\ldots$ consists of finite words, necessarily distinct since they are in distinct $\J$-classes.  By passing to a further subsequence, we may assume that $(w_n)$ converges in $\AL$ to an element $v$ (by compactness of $\AL$). Moreover, since $A^*$ is discrete in $\AL$, and the sequence $w_1,w_2, \ldots$ is a sequence of distinct finite words, we must have that $v\notin A^*$.  As $\geq_{\J}$ is a closed relation on any compact monoid \cite[Proposition~3.1.9]{RS2009}, $v\geq_{\J} w$ and hence $v\mathrel{\J} w$ by maximality. Therefore, $w_1$ is a factor of $v$.  If $w_1=\emp$, then $w_1\geq_{\J} w_2$ and so we may assume that $w_1\neq \emp$.

Let $\phi$ be the formula stating that the finite word $w_1$ is a factor of an $A$-word.  Then since $\phi$ is true for all models of $v$ (because $w_1$ is a factor of $V$ for any $\omega$-saturated model $V$ in the class $v$ by Proposition~\ref{prop:findfact}) and $w_n\to v$, we must have that, for $n$ large enough, $w_n\models \phi$.  Thus we can find $N>1$ with $w_1\geq_{\J} w_N$.  This completes the proof that $w$ is well-factor-ordered.
\end{proof}

Our next aim is to prove (Theorem~\ref{thm:wfosubst}) that a substitution of well-factor-ordered elements into a well-factor-ordered element gives a well-factor-ordered element. The following preliminary proposition will be used in the proof.
\begin{proposition}\label{prop:wfoimplieswsowpo}
Any well-factor-ordered element of $\AL$ is well-suffix-ordered and well-prefix-ordered.
\end{proposition}
\begin{proof}
Let $w \in \AL$. By Lemma~\ref{lem:stabilityetc}.5, the quasi-orders $\leq_{\L}$ and $\leq_{\J}$ coincide on ${\uparrow}_{\L}{w}$. Therefore, $({\uparrow}_{\L}{w},\geq_{\L})$ is a sub-quasi-order of $({\uparrow}_{\J}{w},\geq_{\J})$, and so is $({\uparrow}_{\R}{w},\geq_{\R})$. Sub-quasi-orders of wqo's are wqo's.
\end{proof}

\begin{theorem}\label{thm:wfosubst}
Let $f \colon \BL \to \AL$ be a substitution. If $f(b)$ is well-factor-ordered for each $b \in B$, and $v \in \BL$ is well-factor-ordered, then $f(v)$ is well-factor-ordered.
\end{theorem}
\begin{proof}
Suppose that $f(b)$ is wfo for each $b \in B$ and that $v \in \BL$ is wfo. Let $(w_i)_{i \in \mathbb N}$ be a sequence in ${\uparrow}_{\J}f(v)$. If $w_i = \emp$ for some $i$ then we are done because $\emp \geq_{\J} w$ for any $w$. If there exists $b \in B$ such that $f(b) \leq_{\J} w_i$ infinitely often, then we are done immediately because $f(b)$ is wfo. Now assume that $\emp$ does not occur in $(w_i)_{i \in \omega}$ and that, for each $b \in B$, $f(b) \leq_{\J} w_i$ only finitely often. Passing to a subsequence, we may assume that $f(b) \nleq_{\J} w_i$ for every $i$ and $b \in B$. For each $i$, by Theorem~\ref{thm:substfactors}.1, pick $b_i, c_i \in B$, $x_i, y_i \in \AL$ and $z_i \in \BL$ such that $f(b_i) \leq_{\L} x_i$, $f(c_i) \leq_{\R} y_i$, $v \leq_{\J} b_iz_ic_i$, and $w_i = x_i f(z_i) y_i$. Since the alphabet $B$ is finite, passing to a further subsequence, we may assume that all the $b_i$ are equal to one and the same letter $b$, and that all the $c_i$ are equal to one and the same letter $c$. By Proposition~\ref{prop:wfoimplieswsowpo}, $({\uparrow}_{\L}f(b),\geq_{\L})$ and $({\uparrow}_{\R}f(c),\geq_{\R})$ are wqo. 
In particular, passing to a further subsequence, we get that the sequence $(x_i)_{i \in \mathbb N}$ is non-decreasing in ${({\uparrow}_{\L}f(b),\geq_{\L})}$, $(y_i)_{i \in \omega}$ is non-decreasing in ${({\uparrow}_{\R}f(c),\geq_{\R})}$, and $(bz_ic)_{i \in \omega}$ is non-decreasing in ${({\uparrow}_{\J}v,\geq_{\J})}$.
Thus, in this subsequence, we have in particular that $x_1 \geq_{\L}  x_2 \geq_{\L} f(b)$, ${bz_1c \geq_{\J} bz_2c \geq_{\J} v}$ and $y_1 \geq_{\R} y_2 \geq_{\R} f(c)$. We show that $w_1 \geq_{\J} w_2$. Since $bz_2c \leq_{\J} bz_1c$, by Lemma~\ref{lem:bzc}, there are four cases:

(1) $z_1 = z_2$. In this case,
\[ w_2 = x_2 f(z_2) y_2 = x_2 f(z_1) y_2 \leq_{\J} x_1 f(z_1) y_1 = w_1,\]
using that $x_2 \leq_{\L} x_1$ and $y_2 \leq_{\R} y_1$.

(2) $z_2 \leq_{\R} z_1c$. Then also $f(z_2) \leq_{\R} f(z_1c) = f(z_1)f(c)$, and hence
\[ w_2 \leq_{\R} x_2 f(z_2) \leq_{\R} x_2 f(z_1)f(c) \leq_{\L} x_1 f(z_1)f(c) \leq_{\R} x_1 f(z_1) y_1 = w_1,\]
using that $x_2 \leq_{\L} x_1$ and $f(c) \leq_{\R} y_1$. Thus, $w_1 \leq_{\J} w_2$.

(3) $z_2 \leq_{\L} bz_1$. Analogous to case (2).

(4) $z_2 \leq_{\J} bz_1c$. Then also $f(z_2) \leq_{\J} f(bz_1c) = f(b)f(z_1)f(c)$, and hence
\[ w_2 \leq_{\J} f(z_2) \leq_{\J} f(b) f(z_1) f(c) \leq_{\J} x_1 f(z_1) y_1 = w_1, \]
using that $f(b) \leq_{\L} x_1$ and $f(c) \leq_{\R} y_1$.
\end{proof}

\begin{theorem}\label{thm:wposubst}
Let $f \colon \BL \to \AL$ be a substitution. If $f(b)$ is well-prefix-ordered for each $b \in B$, and $v \in \BL$ is well-prefix-ordered, then $f(v)$ is well-prefix-ordered.
\end{theorem}
\begin{proof}
Let $(w_i)_{i \in \omega}$ be a sequence in ${\uparrow}_{\R}f(v)$. Again, we may assume that $w_i\neq \emp$ for all $i\in \omega$.  By Theorem~\ref{thm:substfactors}.2, for each $i \in \omega$, pick $b_i \in B$, $x_i \in \AL$, and $z_i \in \BL$ such that $f(b_i) \leq_{\R} x_i$, $v \leq_{\R} z_i b_i$, and $w_i = f(z_i)x_i$. As in the proof of Theorem~\ref{thm:wfosubst}, we may pass to a subsequence where all of the $b_i$ are equal to one and the same $b$, $(x_i)_{i \in \omega}$ is non-decreasing in $({\uparrow}f(b),\geq_{\R})$, and $(z_ib)_{i \in \omega}$ is non-decreasing in $({\uparrow}v,\geq_{\R})$. In particular, $x_1 \geq_{\R} x_2 \geq_{\R} f(b)$ and $z_1b \geq_{\R} z_2b \geq_{\R} v$. Pick $\alpha \in \BL$ such that $z_1b\alpha = z_2b$. By Proposition~\ref{prop:equidiv}, pick $\beta \in \BL$ such that either (1) $z_1b\beta = z_2$ and $\beta b = \alpha$, or (2) $z_2\beta = z_1b$ and $\beta \alpha = b$.

(1) We have $f(z_2) \leq_{\R} f(z_1)f(b) \leq_{\R} w_1$, and hence $w_2 = f(z_2)x_2 \leq_{\R} w_1$.

(2) Since $\beta\alpha = b$, we either have $\beta = \emp$ or $\beta = b$. If $\beta = \emp$, then we proceed as in (1). If $\beta = b$, then $z_2b = z_1b$, which implies $z_2 = z_1$ by Corollary~\ref{c:cancel}. Hence,
\[ w_2 = f(z_2)x_2 \leq_{\R} f(z_2)x_1 = f(z_1)x_1 = w_1.\qedhere\]
\end{proof}

By symmetry, we get the following corollary.
\begin{corollary}\label{cor:wsosubst}
Let $f \colon \BL \to \AL$ be a substitution. If $f(b)$ is well-suffix-ordered for each $b \in B$, and $v \in \BL$ is well-suffix-ordered, then $f(v)$ is well-suffix-ordered.\qed
\end{corollary}

The following special case recovers~\cite[Corollary~5.6]{ACZ14} and~\cite[Theorem~7.3]{ACZ14}.

\begin{corollary}\label{cor:omegawqo}
The sets of well-factor-ordered, well-prefix-ordered, and well-suffix-ordered elements in $\AL$ are closed under concatenation and $\omega$-power. In particular, interpretations of $\omega$-terms in $\FA(A)$ are well-factor-ordered.
\end{corollary}
\begin{proof}
The statement about concatenation is immediate from Theorems~\ref{thm:wfosubst}, \ref{thm:wposubst} and Corollary~\ref{cor:wsosubst}.
For the statement about $\omega$-power, notice that $\omega \in \FA(1)$ is well-factor-ordered (and hence also well-suffix-ordered and well-prefix-ordered). The statement about interpretations of $\omega$-terms follows by induction, because finite words are obviously well-factor-ordered.
\end{proof}

It follows that the smallest submonoid of $\FA(A)$ containing all finite words and $\J$-maximal elements that is closed under the $\omega$-power consists of well-factor-ordered elements.  Also, substituting $\omega$-terms into $\J$-maximal elements will provide new example of well-factor-ordered elements.

We call an element $u \in \AL$ \emph{factor-regular} if the set $F(u) := {\uparrow}_{\J}u \cap A^*$ of finite factors of $u$ is a regular language. We call $u$ \emph{prefix-regular} if $P(u) := {\uparrow}_{\R}u \cap A^*$ is a regular language, and \emph{suffix-regular} if $S(u) := {\uparrow}_{\L}u \cap A^*$ is a regular language.  If $u\in \AL\setminus A^*$ is prefix-regular, then so is any element of $u\AL$, and dually for suffix-regular elements. We call a substitution $f \colon \BL \to \AL$ \emph{non-erasing} if $f(b) \neq \emp$ for every $b \in B$. Notice that a non-erasing substitution $f$ sends the ideal $\BL\setminus B^*$ to the ideal $\AL\setminus A^*$.

\begin{theorem}\label{thm:regularity}
Let $f \colon \BL \to \AL$ be a non-erasing substitution and $v \in \BL$.
\begin{enumerate}
\item If $f(b)$ is prefix-regular, for each $b \in B$, and $v$ is prefix-regular, then $f(v)$ is prefix-regular.
\item If $f(b)$ is suffix-regular, for each $b \in B$, and $v$ is suffix-regular, then $f(v)$ is suffix-regular.
\item If $f(b)$ is factor-regular, prefix-regular and suffix-regular, for each $b \in B$, and $v$ is factor-regular, then $f(v)$ is factor-regular.
\end{enumerate}
\end{theorem}
\begin{proof}
Put $C=\{b\in B^*\mid f(b)\in A^*\}$.

(1) Suppose $f(b)$ is prefix-regular for each $b \in B$, and $v$ is prefix-regular. Using Corollary~\ref{cor:substfactors} it is straightforward to verify that
\[P(f(v)) = \bigcup_{b\in B\cup \{\varepsilon\}} f(C^*\cap P(v)b^{-1})P(f(b)).\]
As the regular languages are closed under product, intersection, right quotients and images, we conclude that $P(f(v))$ is regular.

(2) Follows from (1) by symmetry.

(3) Using Corollary~\ref{cor:substfactors} it is straightforward to verify that
\[F(f(v)) = \bigcup_{b\in B\cap F(v)} F(f(b))\cup \bigcup_{b_1,b_2\in B\cup \{\varepsilon\}}S(f(b_1))f(C^*\cap b_1^{-1}F(w)b_2^{-1})P(f(b_1))\] and so the desired result follows from closure of regular languages under boolean operations, product, left and right quotients and homomorphic image.
\end{proof}

For a finite alphabet $A$ and $b \not\in A$, let $f \colon \Lambda(A \cup \{b\}) \to \AL$ be the substitution erasing $b$ and fixing $A$. Note that if $v\in \AL$ is not prefix-regular, then $b^{\omega}v$ is prefix-regular, but $f(b^{\omega}v)=v$ is not.  Thus the assumption in Theorem~\ref{thm:regularity} that $f$ is non-erasing is necessary.  The next corollary recovers~\cite[Corollary~7.6]{ACZ14}.

\begin{corollary}
Interpretations of $\omega$-terms are prefix-, suffix- and factor-regular.
\end{corollary}

\begin{remark}
The minimal ideal $I$ of $\FA(A)$ consists of those elements containing every finite word as a factor. In particular, every element of $I$ is factor-regular.  Thus if we substitute $\omega$-terms over $B$ into an element of the minimal ideal of $\FA(A)$, then we obtain factor-regular elements of $\FA(B)$.  It is not the case that every element of the minimal ideal of $\FA(A)$ is prefix-regular or suffix-regular.  In fact, it is easy to see using the pumping lemma that $w\in \AL\setminus A^*$ is prefix-regular if and only if $w=uv^{\omega}z$ with $u,v$ finite and $v\neq \emp$.  A dual description holds for suffix-regular elements.
\end{remark}

\section*{Conclusion and further work}
In this paper we gave a new approach to the free pro-aperiodic monoid by viewing its elements as elementary equivalence classes of pseudofinite words. This view led us to consider $\FA(A)$ as a topologically closed submonoid of the larger monoid $\AL$ consisting of elementary equivalence classes of arbitrary $A$-words. 
The model-theoretic fact that each such class contains an $\omega$-saturated model enabled us to analyze factors in $\AL$ combinatorially. Thus, we substantiate the claim made in \cite{HRS2010} that one may ``transfer arguments from Combinatorics on Words to the profinite context'': our approach using saturated models makes this idea precise.

The newly identified pro-aperiodic monoid $\AL$ poses several interesting questions for future work. In particular, it would be interesting to study the algebraic structure of $\AL$ in more detail. Here, connections with the work of Carton, Colcombet and Puppis on algebras for words over countable linear orderings \cite{CCP2011} are to be expected.

We also plan to explore in future work how this method might be extended to other profinite monoids. In particular, one could try to analyze the absolutely free profinite monoid in this way, but this would require replacing first-order model theory by monadic second-order model theory. In this direction, we foresee connections to Shelah's seminal work \cite{Shelah1975}.

In a different direction, we hope that our approach could be useful for more easily analyzing aperiodic pointlike sets and related notions, in particular because a logical approach has recently proved useful for deciding problems in the first-order quantifier alternation hierarchy \cite{PlaZei2014}.

\subsection*{Historical remark}
The second-named author first began to think about using model theory to study free pro-aperiodic monoids in late 2008 and suggested to his PhD student, David Gains, at Carleton University to work on axiomatizing the theory of finite words, describing the multiplication on pseudofinite words and proving equidivisibility from this viewpoint.  Unfortunately, the project did not have a chance to go far before the second-named author left Carleton.  The authors recommenced the project in January 2016 and proved the results presented here.  
As we began to give talks on this work, we were informed by several researchers that they had also been thinking for some time about applying model theory to studying free pro-aperiodic monoids,  and had some unpublished results that could be related to ours. In particular, the preprint \cite{ACCZ2017} was published on arXiv in February 2017, after we published our first paper on the topic on arXiv in September 2016, which later appeared in the conference STACS 2017 \cite{GS2016STACS}. The work in \cite{ACCZ2017} is related to, but different from, our work here, one important difference being our use of saturated models.

\subsection*{Acknowledgements} The first-named author was supported by the European Union's Horizon 2020 research and innovation programme under the Marie Sklodowska-Curie grant \#655941; the second-named author was supported by Simons Foundation \#245268, United States - Israel Binational Science Foundation \#2012080 and NSA MSP \#H98230-16-1-0047.

\appendix

\section{Model theory}\label{sec:modeltheory}
In this section, we provide justification for the model-theoretic terminology and results that were used throughout the paper.

\subsection{Relativizing and parameters}
A key fact that is specific to words is that formulas can be relativized to intervals and rays, in the following sense (cf., e.g., \cite[Lem.~VI.1.3]{Str1994} and \cite[Def.~13.27]{Ros1982}.)
\begin{lemma}\label{lem:relativize}
For any formula $\phi(\x)$, there exist formulas $\phi^{(y,z)}(\x,y,z)$, $\phi^{{<}y}(\x,y)$ and $\phi^{{>}y}(\x,y)$ such that, for any $A$-word $W$ and $i, j \in |W|$,
\begin{itemize}[leftmargin=*]
\item $W, \x^W i j \models \phi^{(y,z)}$ if and only if  $W(i,j), \x^W \models \phi$, for any $\x^W$ in $|W(i,j)|$.
\item $W, \x^W i \models \phi^{{<}y}$ if and only if  $W\be{i},\x^W \models \phi$, for any $\x^W$ in $|W\be{i}|$.
\item $W, \x^W i \models \phi^{{>}y}$ if and only if  $W\af{i},\x^W \models \phi$, for any $\x^W$ in $|W\af{i}|$.
\end{itemize}
\end{lemma}
\begin{proof}
The formulas $\phi^{(y,z)}$, $\phi^{{<}y}$ and $\phi^{{>}y}$ are defined by induction on the complexity of $\phi$. We only give the definitions for $\phi^{(y,z)}$, leaving the rest of the proof to the reader. For any atomic formula $\phi$, define $\phi^{(y,z)}$ to be $\phi$. For any Boolean combination $\phi$ of formulas $\phi_1$, $\phi_2$, define $\phi^{(y,z)}$ to be the same Boolean combination of $\phi_1^{(y,z)}$ and $\phi_2^{(y,z)}$. Finally, if $\phi$ is of the form $\exists x_i \psi$ or $\forall x_i \psi$, define $\phi^{(y,z)}$ by, respectively,
\begin{align*}
&\exists x_i (y < x_i \wedge x_i < z \wedge \psi^{(y,z)}), \text{ or }\\
&\forall x_i ( (y < x_i \wedge x_i < z) \to \psi^{(y,z)}).\qedhere
\end{align*}
\end{proof}

In model theory, one considers a more general version of the equivalence relations $\equiv_k$ introduced in Section~\ref{sec:logic}. For every $n,k \geq 0$, there is an equivalence relation $\equiv_{n,k}$ between models with assignments of $n$ first-order variables: two such models with assignments, $(U,\overline{i})$ and $(V,\overline{j})$, are \emph{$\equiv_{n,k}$-equivalent} if they satisfy the same formulas $\phi(x_1,\dots,x_n)$ of quantifier depth $\leq k$; notation $U,\overline{i} \equiv_{n,k} V,\overline{j}$. In this context, the tuples $\overline{i}$ and $\overline{j}$ are called \emph{parameters}. We call the models with parameters $(U,\overline{i})$ and $(V,\overline{j})$ \emph{elementarily equivalent} if $U,\overline{i} \equiv_{n,k} V,\overline{j}$ for every $k \geq 0$; notation $U,\overline{i} \equiv V,\overline{j}$.

The \emph{$k$-round EF-game} on $(U,\overline{i})$ and $(V,\overline{j})$, where $\overline{i} \in U^n$ and $\overline{j} \in V^n$, is played exactly as the one without parameters, but the winning condition now requires that the map defined by $u_p \mapsto v_p$ and $i_q \mapsto j_q$ is an isomorphism between the substructure $\{u_p \ | \ 1 \leq p \leq k\} \cup \{i_q \ | \ 1 \leq q \leq n\}$ of $U$ and the substructure $\{v_p \ | \ 1 \leq p \leq k\} \cup \{j_q \ | \ 1 \leq q \leq n\}$ of $V$.
In the case of $A$-words, the equivalence relation $\equiv_{n,k}$ can be reduced to the relation $\equiv_k$ holding between intervals of the models, as follows. 
For an $n$-tuple $\overline{i}$ in an $A$-word $U$, we say an open ray or interval with endpoints in $\overline{i}$ is \emph{$\overline{i}$-minimal} if it does not contain any elements from $\overline{i}$.
\begin{proposition}\label{prop:EFgamesparam}
Let $U$, $V$ be $A$-words, $\overline{i} \in |U|^n$ and $\overline{j} \in |V|^n$. The following are equivalent:
\begin{enumerate}
\item $U, \overline{i} \equiv_{n,k} V, \overline{j}$;
\item $U(i_p) = V(j_p)$ for all $1 \leq p \leq n$, the order $<^U$ on $\overline{i}$ coincides with the order $<^V$ on $\overline{j}$, and every $\overline{i}$-minimal open interval or ray in $U$ with endpoints in $\overline{i}$ is $k$-equivalent to the open interval or ray in $V$ with corresponding endpoints in $\overline{j}$.
\end{enumerate}
\end{proposition}
\begin{proof}
For (1) implies (2), we have the same letters at $i_p$ and $j_p$ because each $P_a(x)$ is a formula of quantifier depth $0$, and the orders coincide because $x_p < x_q$ is of quantifier depth $0$. If $(i_p,i_q)$ and $(j_p,j_q)$ are corresponding minimal open intervals, and $U(i_p,i_q) \models \phi$ for some formula $\phi$ of quantifier depth $k$, then $U, \overline{i} \models \phi^{(x_p,x_q)}$, so, by assumption, $V, \overline{j} \models \phi^{(x_p,x_q)}$, and hence $V(j_p,j_q) \models \phi$. So $U(i_p,i_q) \equiv_k V(j_p,j_q)$. The proof for minimal open rays is analogous.

For (2) implies (1), we describe a winning strategy for player $\exists$ in the $k$-round game on $(U,\overline{i})$ and $(V,\overline{j})$. If $\forall$ plays one of the $i_p$ and $j_p$, $\exists$ responds with the corresponding parameter in the other model. If $\forall$ plays in $U$, say, at a position $i'$ not in $\overline{i}$, then there is a unique $\overline{i}$-minimal open interval or ray with endpoints in $\overline{i}$ that contains $i'$. The winning strategy for the $k$-round game on this $\overline{i}$-minimal open interval or ray then gives an element for $\exists$ to play in the corresponding open interval or ray in $V$. Since the strategies for $\exists$ on the $\overline{i}$-minimal open intervals and rays are winning by assumption, and the orders on the tuples $\overline{i}$ and $\overline{j}$ coincide, the resulting submodels of $U$ and $V$ are isomorphic.
\end{proof}
Note that a special case of Proposition~\ref{prop:EFgamesparam} is that $i$ $k$-corresponds to $j$ if and only if  $U, i \equiv_{1,k} V, j$. Also note that Proposition~\ref{prop:EFgamesparam} implies in particular that for any $A$-words $U_1, \dots, U_m$ and $V_1, \dots, V_m$ such that $U_i \equiv V_i$ for $i = 1, \dots, m$,  we have $U_1 \cdot \cdots \cdot U_m \equiv_k V_1 \cdot \cdots \cdot V_m$, which is a special case of Proposition~\ref{prop:substinv-k}.

\subsection{Comparison with notions from model theory}
We defined the \emph{type} $t^U(i)$ of a position $i$ in an $A$-word $U$ to be $([U\be{i}]_{\equiv},U(i),[U\af{i}]_{\equiv})$. This corresponds to what is called a \emph{complete $1$-type with respect to an empty set of parameters} in model theory, which is usually defined as the set of those formulas $\theta(x)$ such that $U,i \models \theta(x)$. Indeed, $t^U(i)$ contains precisely the same information as this set, as the following proposition shows. 
\begin{proposition}\label{prop:onetypes}
For any $A$-words $U, V$, $i \in |U|$ and $j \in |V|$, the following are equivalent:
\begin{enumerate}
\item $t^U(i) = t^V(j)$;
\item for every $k \geq 0$, $i$ $k$-corresponds to $j$;
\item $U,i \equiv V,j$.
\end{enumerate}
\end{proposition}
\begin{proof}
(1) implies (2) is clear from the definitions. By the remark following Proposition~\ref{prop:EFgamesparam}, (2) implies that $U, i \equiv_{1,k} V, j$ for all $k$, which gives (3). For (3) implies (1), note first that we have $U(i) = V(j)$ using the formulas $P_a(x)$. Also, if $U\be{i} \models \phi$ for some sentence $\phi$, then $U,i \models \phi^{<x}(x)$ by Lemma~\ref{lem:relativize}, so $V, j \models \phi^{<x}(x)$ by (3), and hence $V\be{j} \models \phi$ by Lemma~\ref{lem:relativize}. Thus, $U\be{i} \equiv V\be{j}$. Similarly, $U\af{i} \equiv V\af{j}$. Thus, $t^U(i) = t^V(j)$.
\end{proof}

By Proposition~\ref{prop:onetypes}, a model is \emph{weakly saturated} in our sense if and only if  the model realizes all the complete $1$-types without parameters that are consistent with it. We now prove that what we called \emph{$\omega$-saturated} above is equivalent to the usual model-theoretic notion by the same name; indeed, in Proposition~\ref{prop:omegasatequiv} below, (1) is our definition of $\omega$-saturated, and (2) is well-known to be equivalent to the standard model-theoretic definition of $\omega$-saturated, cf.~\cite[Proposition~4.3.2]{Mar2002}.

\begin{proposition}\label{prop:omegasatequiv}
Let $U$ be an $A$-word. The following are equivalent:
\begin{enumerate}
\item Every closed interval in $U$ is weakly saturated.
\item For any $n$-tuple $\overline{i}$ in $U$ and for any model with $n$ parameters $(V,\overline{j})$ such that $U,\overline{i} \equiv V,\overline{j}$, if $j' \in |V|$, then there exists $i' \in |U|$ such that $U,\overline{i}i' \equiv V,\overline{j}j'$.
\end{enumerate}
\end{proposition}
\begin{proof}
Suppose that (1) holds. Let $(U,\overline{i}) \equiv (V,\overline{j})$ and let $j' \in |V|$. If $j'$ is in $\overline{j}$, then pick the corresponding $i'$ in $\overline{i}$. Otherwise, pick the $\overline{j}$-minimal open interval $(j_p,j_q)$ that contains $j'$ (the proof is the same if $j'$ lies in a $\overline{j}$-minimal open ray $V\be{j_p}$ or $V\af{j_q}$). By Proposition~\ref{prop:EFgamesparam}, $U(i_p,i_q) \equiv V(j_p,j_q)$. By (1), $U(i_p,i_q)$ is weakly saturated (using that every open interval occurs as a closed interval since the order is discrete), so there exists $i' \in (i_p,i_q)$ that realizes the same type as $j'$. By Proposition~\ref{prop:EFgamesparam}, we obtain $U,\overline{i}i' \equiv V,\overline{j}j'$.

Suppose that (2) holds. Let $[i_1,i_2]$ be a closed interval in $U$. We show that $U[i_1,i_2]$ is weakly saturated. Let $V' \equiv U[i_1,i_2]$ and $j' \in |V'|$. Let $V$ be the $A$-word $U\be{i_1} \cdot V' \cdot U\af{i_2}$, and denote by $j_1$ and $j_2$ the first and last position of $V'$ in $V$. By Proposition~\ref{prop:EFgamesparam}, $V,j_1j_2 \equiv U, i_1i_2$. By (2), pick $i'$ in $U$ such that $U, i_1i_2i' \equiv V,j_1j_2j'$. Since $j' \in [j_1,j_2]$ in $V$, we must have $i' \in [i_1,i_2]$ in $U$. Using Proposition~\ref{prop:EFgamesparam}, we get $U[i_1,i_2],i' \equiv V[j_1,j_2],j'$, and $V[j_1,j_2] = V'$ by definition. Thus, $t^{U[i_1,i_2]}(i') = ([V'\be{j'}]_{\equiv},V'(j'),[V'\af{j'}]_{\equiv}) = t^{V'}(j')$, so the type of $j'$ in $V'$ is realized in $U[i_1,i_2]$, as required.
\end{proof}

The next proposition now follows from \cite[Cor.~10.2.2]{Hod1993} or \cite[Thm.~4.3.12]{Mar2002}.
\begin{proposition}\label{prop:satexist}
For any $u \in \AL$, there exists an $\omega$-saturated $A$-word $U$ in the elementary equivalence class $u$.
\end{proposition}

\bibliographystyle{amsplain}

\def\cprime{$'$} \def\cprime{$'$}
\providecommand{\bysame}{\leavevmode\hbox to3em{\hrulefill}\thinspace}
\providecommand{\MR}{\relax\ifhmode\unskip\space\fi MR }
\providecommand{\MRhref}[2]{%
  \href{http://www.ams.org/mathscinet-getitem?mr=#1}{#2}
}
\providecommand{\href}[2]{#2}
\bibliography{goolsteinberg}

\end{document}